\documentclass[submission,copyright,creativecommons,sharealike,noncommercial]{eptcs}
\providecommand{\event}{New Journal of Physics}

\usepackage{graphicx}
\usepackage{amsfonts}
\usepackage{amsmath}  
\usepackage{amsthm}
\usepackage{amssymb}
\usepackage{relsize}
\usepackage{mathtools}
\usepackage[applemac]{inputenc} 
\usepackage[nocompress]{cite}

\usepackage[usenames,dvipsnames]{xcolor} 
\usepackage{tikz}
\usepackage{tikzfig}

\usetikzlibrary{arrows,shapes,decorations,backgrounds,positioning}

	\newcounter{theorem_c} 
	\numberwithin{theorem_c}{section} 
	\numberwithin{equation}{section} 

	\theoremstyle{plain} 
	\newtheorem{theorem}[theorem_c]{Theorem}
	
	\newtheorem{lemma}[theorem_c]{Lemma}
	\newtheorem{corollary}[theorem_c]{Corollary}

	\theoremstyle{definition} 
	\newtheorem{definition}[theorem_c]{Definition}
	\newtheorem{example}[theorem_c]{Example}

	\newcommand\numberthis{\addtocounter{equation}{1}\tag{\theequation}} 

	\newcommand{\integers}{\mathbb{Z}} 
	\newcommand{\reals}{\mathbb{R}} 
	\newcommand{\integersMod}[1]{\mathbb{Z}_{#1}} 
	\newcommand{\modclass}[2]{#1 \; (\text{mod } #2)} 

	\newcommand{\suchthat}[2]{\left\{#1 \: \text{ s.t. } \: #2\right\}} 

		\newcommand{\ket}[1]{\vert #1 \rangle} 
		\newcommand{\bra}[1]{\langle #1 \vert} 
		\newcommand{\braket}[2]{\langle #1 \vert #2 \rangle} 

		\newcommand{\SpaceH}{\mathcal{H}}

		\newcommand{\isom}{\cong} 

		\newcommand{\tensor}{\otimes} 

		\newcommand{\fdHilbCategory}{\operatorname{fdHilb}} 
		\newcommand{\fRelCategory}{\operatorname{fRel}} 

		\newcommand{\CategoryC}{\mathcal{C}}

	\newcommand{\Xcolour}{Red}
	\newcommand{\groupStructColour}{\Xcolour}
	\newcommand{\XmultSym}{\hbox{\begin{tikzpicture} [scale=1,transform shape] 

\def\deltax{0.3} 
\def\deltay{0.5} 

\node (mult_label_inl) at (-\deltax,-\deltay) {};
\node (mult_label_inr) at (+\deltax,-\deltay) {};
\node [dot, fill=\groupStructColour] (mult) at (0,0) {};
\node (mult_label_out) at (0,+\deltay) {};

\draw[-] [out=90,in=225](mult_label_inl) to (mult);
\draw[-] [out=90,in=315](mult_label_inr) to (mult);
\draw[-] (mult) to (mult_label_out);

\end{tikzpicture}}\!}

	\newcommand{\XcounitSym}{\hbox{\begin{tikzpicture} [scale=1,transform shape] 

\def\deltax{0.3} 
\def\deltay{0.5} 

\node [dot, fill=\groupStructColour] (mult) at (0,0) {};
\node (mult_label_in) at (0,-\deltay) {};
\draw[-] (mult_label_in) to (mult);

\end{tikzpicture}}\!}

	\newcommand{\Zcolour}{YellowGreen}
	\newcommand{\classicalStructColour}{\Zcolour}
	\newcommand{\ZmultSym}{\hbox{\begin{tikzpicture} [scale=1,transform shape] 

\def\deltax{0.3} 
\def\deltay{0.5} 

\node (mult_label_inl) at (-\deltax,-\deltay) {};
\node (mult_label_inr) at (+\deltax,-\deltay) {};
\node [dot, fill=\classicalStructColour] (mult) at (0,0) {};
\node (mult_label_out) at (0,+\deltay) {};
\draw[-] [out=90,in=225](mult_label_inl) to (mult);
\draw[-] [out=90,in=315](mult_label_inr) to (mult);
\draw[-] (mult) to (mult_label_out);

\end{tikzpicture}}\!}
	
	\newcommand{\ZunitSym}{\hbox{\begin{tikzpicture} [scale=1,transform shape] 

\def\deltax{0.3} 
\def\deltay{0.5} 

\node [dot, fill=\classicalStructColour] (mult) at (0,0) {};
\node (mult_label_out) at (0,+\deltay) {};
\draw[-] (mult) to (mult_label_out);

\end{tikzpicture}}\!}

	\tikzset{->-/.style={decoration={markings,mark=at position #1 with {\arrow{>}}},postaction={decorate}}}
	\tikzset{-<-/.style={decoration={markings,mark=at position #1 with {\arrow{<}}},postaction={decorate}}}

\def\swangle{-145}
\def\seangle{-35}
\def\nwangle{145}
\def\neangle{35}

\tikzstyle{env}=[copoint,regular polygon rotate=0,minimum width=0.2cm, fill=black]

\tikzstyle{probs}=[shape=semicircle,fill=white,draw=black,shape border rotate=180,minimum width=1.2cm]

\tikzstyle{every picture}=[baseline=-0.25em,scale=0.5]
\tikzstyle{dotpic}=[] 
\tikzstyle{diredges}=[every to/.style={diredge}]
\tikzstyle{math matrix}=[matrix of math nodes,left delimiter=(,right delimiter=),inner sep=2pt,column sep=1em,row sep=0.5em,nodes={inner sep=0pt},text height=1.5ex, text depth=0.25ex]

\tikzstyle{inline text}=[text height=1.5ex, text depth=0.25ex,yshift=0.5mm]
\tikzstyle{label}=[font=\footnotesize,text height=1.5ex, text depth=0.25ex,yshift=0.5mm]
\tikzstyle{left label}=[label,anchor=east,xshift=1.5mm]
\tikzstyle{right label}=[label,anchor=west,xshift=-1.5mm]

\tikzstyle{braceedge}=[decorate,decoration={brace,amplitude=2mm,raise=-1mm}]
\tikzstyle{small braceedge}=[decorate,decoration={brace,amplitude=1mm,raise=-1mm}]

\tikzstyle{doubled}=[line width=1.6pt] 
\tikzstyle{boldedge}=[doubled,shorten <=-0.17mm,shorten >=-0.17mm]
\tikzstyle{boldedgegray}=[doubled,gray,shorten <=-0.17mm,shorten >=-0.17mm]

\tikzstyle{semidoubled}=[line width=1.4pt] 
\tikzstyle{semiboldedgegray}=[semidoubled,gray,shorten <=-0.17mm,shorten >=-0.17mm]

\tikzstyle{boldedgedashed}=[very thick,dashed,shorten <=-0.17mm,shorten >=-0.17mm]
\tikzstyle{vboldedgedashed}=[doubled,dashed,shorten <=-0.17mm,shorten >=-0.17mm]
\tikzstyle{left hook arrow}=[left hook-latex]
\tikzstyle{right hook arrow}=[right hook-latex]
\tikzstyle{sembracket}=[line width=0.5pt,shorten <=-0.07mm,shorten >=-0.07mm]

\tikzstyle{causal edge}=[->,thick,gray]
\tikzstyle{causal nondir}=[thick,gray]
\tikzstyle{timeline}=[thick,gray, dashed]

\tikzstyle{cedge}=[<->,thick,gray!70!white]

\tikzstyle{empty diagram}=[draw=gray!40!white,dashed,shape=rectangle,minimum width=1cm,minimum height=1cm]
\tikzstyle{empty diagram small}=[draw=gray!50!white,dashed,shape=rectangle,minimum width=0.6cm,minimum height=0.5cm]

\tikzstyle{dot}=[inner sep=0mm,minimum width=3mm,minimum height=3mm,draw,shape=circle,text depth=-0.1mm]
\tikzstyle{smalldot}=[inner sep=0mm,minimum width=2mm,minimum height=2mm,draw,shape=circle,text depth=-0.1mm]
\tikzstyle{ddot}=[inner sep=0mm, doubled, minimum width=3.5mm,minimum height=3.5mm,draw,shape=circle]

\tikzstyle{black dot}=[dot,fill=black]
\tikzstyle{white dot}=[dot,fill=white,,text depth=-0.2mm]
\tikzstyle{green dot}=[white dot] 
\tikzstyle{gray dot}=[dot,fill=gray!40!white,,text depth=-0.2mm]
\tikzstyle{red dot}=[gray dot] 

\tikzstyle{black ddot}=[ddot,fill=black]
\tikzstyle{white ddot}=[ddot,fill=white]
\tikzstyle{gray ddot}=[ddot,fill=gray!40!white]

\tikzstyle{gray edge}=[gray!40!white]

\tikzstyle{small dot}=[inner sep=0.5mm,minimum width=0pt,minimum height=0pt,draw,shape=circle]

\tikzstyle{small black dot}=[small dot,fill=black]
\tikzstyle{small white dot}=[small dot,fill=white]
\tikzstyle{small gray dot}=[small dot,fill=gray!40!white]

\tikzstyle{causal dot}=[inner sep=0.4mm,minimum width=0pt,minimum height=0pt,draw=white,shape=circle,fill=gray!40!white]

\tikzstyle{phase dimensions}=[minimum size=5mm,font=\footnotesize,rectangle,rounded corners=2.5mm,inner sep=0.2mm,outer sep=-2mm,text height=1ex, text depth=0.25ex, yshift=0.5mm]
\tikzstyle{dphase dimensions}=[phase dimensions]

\tikzstyle{phase dot}=[dot,phase dimensions]

\tikzstyle{white phase dot}=[dot,fill=white,phase dimensions]
\tikzstyle{white phase ddot}=[ddot,fill=white,dphase dimensions]

\tikzstyle{white rect ddot}=[draw=black,fill=white,doubled,minimum size=5mm,font=\footnotesize,rectangle,rounded corners=2.5mm,inner sep=0.2mm]
\tikzstyle{gray rect ddot}=[draw=black,fill=gray!40!white,doubled,minimum size=6mm,font=\footnotesize,rectangle,rounded corners=3mm]

\tikzstyle{gray phase dot}=[dot,fill=gray!40!white,phase dimensions]
\tikzstyle{gray phase ddot}=[ddot,fill=gray!40!white,dphase dimensions]
\tikzstyle{grey phase dot}=[gray phase dot]
\tikzstyle{grey phase ddot}=[gray phase ddot]

\tikzstyle{cnot}=[fill=white,shape=circle,inner sep=-1.4pt]
\tikzstyle{hadamard}=[square box,inner sep=0 pt,font=\footnotesize,minimum height=4mm,minimum width=4mm]
\tikzstyle{dhadamard}=[hadamard,doubled]
\tikzstyle{antipode}=[white dot,inner sep=0.3mm,font=\footnotesize]

\tikzstyle{scalar}=[diamond,draw,inner sep=0.5pt,font=\small]
\tikzstyle{dscalar}=[diamond,doubled, draw,inner sep=0.5pt,font=\small]

\tikzstyle{small box}=[rectangle,inline text,fill=white,draw,minimum height=5mm,yshift=-0.5mm,minimum width=5mm,font=\small]
\tikzstyle{small gray box}=[small box,fill=gray!30]
\tikzstyle{medium box}=[rectangle,inline text,fill=white,draw,minimum height=5mm,yshift=-0.5mm,minimum width=10mm,font=\small]
\tikzstyle{square box}=[small box] 
\tikzstyle{medium gray box}=[small box,fill=gray!30]
\tikzstyle{semilarge box}=[rectangle,inline text,fill=white,draw,minimum height=5mm,yshift=-0.5mm,minimum width=12.5mm,font=\small]
\tikzstyle{large box}=[rectangle,inline text,fill=white,draw,minimum height=5mm,yshift=-0.5mm,minimum width=15mm,font=\small]
\tikzstyle{large gray box}=[small box,fill=gray!30]

\tikzstyle{gray square point}=[small box,fill=gray!50]

\tikzstyle{dphase box white}=[dbox]
\tikzstyle{dphase box gray}=[dbox,fill=gray!50!white]

\tikzstyle{point}=[regular polygon,regular polygon sides=3,draw,scale=0.75,inner sep=-0.5pt,minimum width=9mm,fill=white,regular polygon rotate=180]
\tikzstyle{copoint}=[regular polygon,regular polygon sides=3,draw,scale=0.75,inner sep=-0.5pt,minimum width=9mm,fill=white]
\tikzstyle{dpoint}=[point,doubled]
\tikzstyle{dcopoint}=[copoint,doubled]

\tikzstyle{wide copoint}=[fill=white,draw,shape=isosceles triangle,shape border rotate=90,isosceles triangle stretches=true,inner sep=0pt,minimum width=1.5cm,minimum height=6.12mm]
\tikzstyle{wide point}=[fill=white,draw,shape=isosceles triangle,shape border rotate=-90,isosceles triangle stretches=true,inner sep=0pt,minimum width=1.5cm,minimum height=6.12mm,yshift=-0.0mm]
\tikzstyle{wide point plus}=[fill=white,draw,shape=isosceles triangle,shape border rotate=-90,isosceles triangle stretches=true,inner sep=0pt,minimum width=1.74cm,minimum height=7mm,yshift=-0.0mm]

\tikzstyle{wide dpoint}=[fill=white,doubled,draw,shape=isosceles triangle,shape border rotate=-90,isosceles triangle stretches=true,inner sep=0pt,minimum width=1.5cm,minimum height=6.12mm,yshift=-0.0mm]

\tikzstyle{tinypoint}=[regular polygon,regular polygon sides=3,draw,scale=0.55,inner sep=-0.15pt,minimum width=6mm,fill=white,regular polygon rotate=180] 

\tikzstyle{white point}=[point]
\tikzstyle{white dpoint}=[dpoint]
\tikzstyle{green point}=[white point] 
\tikzstyle{white copoint}=[copoint]
\tikzstyle{gray point}=[point,fill=gray!40!white]
\tikzstyle{gray dpoint}=[gray point,doubled]
\tikzstyle{red point}=[gray point] 
\tikzstyle{gray copoint}=[copoint,fill=gray!40!white]
\tikzstyle{gray dcopoint}=[gray copoint,doubled]

\tikzstyle{black point}=[point,fill=black]
\tikzstyle{black copoint}=[copoint,fill=black]

\tikzstyle{tiny gray point}=[tinypoint,fill=gray!40!white]

\tikzstyle{diredge}=[->]
\tikzstyle{rdiredge}=[<-]
\tikzstyle{thickdiredge}=[->, very thick]
\tikzstyle{pointer edge}=[->,very thick,gray]
\tikzstyle{pointer edge part}=[very thick,gray]
\tikzstyle{dashed edge}=[dashed]
\tikzstyle{thick dashed edge}=[very thick,dashed]
\tikzstyle{thick gray dashed edge}=[thick dashed edge,gray!40]
\tikzstyle{thick map edge}=[very thick,|->]

\makeatletter
\newcommand{\boxshape}[3]{%
\pgfdeclareshape{#1}{
\inheritsavedanchors[from=rectangle] 
\inheritanchorborder[from=rectangle]
\inheritanchor[from=rectangle]{center}
\inheritanchor[from=rectangle]{north}
\inheritanchor[from=rectangle]{south}
\inheritanchor[from=rectangle]{west}
\inheritanchor[from=rectangle]{east}
\backgroundpath{
\southwest \pgf@xa=\pgf@x \pgf@ya=\pgf@y
\northeast \pgf@xb=\pgf@x \pgf@yb=\pgf@y

\@tempdima=#2
\@tempdimb=#3

\pgfpathmoveto{\pgfpoint{\pgf@xa - 5pt + \@tempdima}{\pgf@ya}}
\pgfpathlineto{\pgfpoint{\pgf@xa - 5pt - \@tempdima}{\pgf@yb}}
\pgfpathlineto{\pgfpoint{\pgf@xb + 5pt + \@tempdimb}{\pgf@yb}}
\pgfpathlineto{\pgfpoint{\pgf@xb + 5pt - \@tempdimb}{\pgf@ya}}
\pgfpathlineto{\pgfpoint{\pgf@xa - 5pt + \@tempdima}{\pgf@ya}}
\pgfpathclose
}
}}

\boxshape{NEbox}{0pt}{5pt}
\boxshape{SEbox}{0pt}{-5pt}
\boxshape{NWbox}{5pt}{0pt}
\boxshape{SWbox}{-5pt}{0pt}
\boxshape{EBox}{-3pt}{3pt}
\boxshape{WBox}{3pt}{-3pt}
\makeatother

\tikzstyle{cloud}=[shape=cloud,draw,minimum width=1.5cm,minimum height=1.5cm]

\tikzstyle{map}=[draw,shape=NEbox,inner sep=2pt,minimum height=6mm,fill=white]
\tikzstyle{dashedmap}=[draw,dashed,shape=NEbox,inner sep=2pt,minimum height=6mm,fill=white]
\tikzstyle{mapdag}=[draw,shape=SEbox,inner sep=2pt,minimum height=6mm,fill=white]
\tikzstyle{mapadj}=[draw,shape=SEbox,inner sep=2pt,minimum height=6mm,fill=white]
\tikzstyle{maptrans}=[draw,shape=SWbox,inner sep=2pt,minimum height=6mm,fill=white]
\tikzstyle{mapconj}=[draw,shape=NWbox,inner sep=2pt,minimum height=6mm,fill=white]

\tikzstyle{medium map}=[draw,shape=NEbox,inner sep=2pt,minimum height=6mm,fill=white,minimum width=7mm]
\tikzstyle{medium map dag}=[draw,shape=SEbox,inner sep=2pt,minimum height=6mm,fill=white,minimum width=7mm]
\tikzstyle{medium map adj}=[draw,shape=SEbox,inner sep=2pt,minimum height=6mm,fill=white,minimum width=7mm]
\tikzstyle{medium map trans}=[draw,shape=SWbox,inner sep=2pt,minimum height=6mm,fill=white,minimum width=7mm]
\tikzstyle{medium map conj}=[draw,shape=NWbox,inner sep=2pt,minimum height=6mm,fill=white,minimum width=7mm]
\tikzstyle{semilarge map}=[draw,shape=NEbox,inner sep=2pt,minimum height=6mm,fill=white,minimum width=9.5mm]
\tikzstyle{semilarge map trans}=[draw,shape=SWbox,inner sep=2pt,minimum height=6mm,fill=white,minimum width=9.5mm]
\tikzstyle{semilarge map adj}=[draw,shape=SEbox,inner sep=2pt,minimum height=6mm,fill=white,minimum width=9.5mm]
\tikzstyle{semilarge map dag}=[draw,shape=SEbox,inner sep=2pt,minimum height=6mm,fill=white,minimum width=9.5mm]
\tikzstyle{semilarge map conj}=[draw,shape=NWbox,inner sep=2pt,minimum height=6mm,fill=white,minimum width=9.5mm]
\tikzstyle{large map}=[draw,shape=NEbox,inner sep=2pt,minimum height=6mm,fill=white,minimum width=12mm]
\tikzstyle{very large map}=[draw,shape=NEbox,inner sep=2pt,minimum height=6mm,fill=white,minimum width=17mm]

\tikzstyle{medium dmap}=[draw,doubled,shape=NEbox,inner sep=2pt,minimum height=6mm,fill=white,minimum width=7mm]
\tikzstyle{medium dmap dag}=[draw,doubled,shape=SEbox,inner sep=2pt,minimum height=6mm,fill=white,minimum width=7mm]
\tikzstyle{medium dmap adj}=[draw,doubled,shape=SEbox,inner sep=2pt,minimum height=6mm,fill=white,minimum width=7mm]
\tikzstyle{medium dmap trans}=[draw,doubled,shape=SWbox,inner sep=2pt,minimum height=6mm,fill=white,minimum width=7mm]
\tikzstyle{medium dmap conj}=[draw,doubled,shape=NWbox,inner sep=2pt,minimum height=6mm,fill=white,minimum width=7mm]
\tikzstyle{semilarge dmap}=[draw,doubled,shape=NEbox,inner sep=2pt,minimum height=6mm,fill=white,minimum width=9.5mm]
\tikzstyle{semilarge dmap trans}=[draw,doubled,shape=SWbox,inner sep=2pt,minimum height=6mm,fill=white,minimum width=9.5mm]
\tikzstyle{semilarge dmap adj}=[draw,doubled,shape=SEbox,inner sep=2pt,minimum height=6mm,fill=white,minimum width=9.5mm]
\tikzstyle{semilarge dmap dag}=[draw,doubled,shape=SEbox,inner sep=2pt,minimum height=6mm,fill=white,minimum width=9.5mm]
\tikzstyle{semilarge dmap conj}=[draw,doubled,shape=NWbox,inner sep=2pt,minimum height=6mm,fill=white,minimum width=9.5mm]
\tikzstyle{large dmap}=[draw,doubled,shape=NEbox,inner sep=2pt,minimum height=6mm,fill=white,minimum width=12mm]
\tikzstyle{large dmap conj}=[draw,doubled,shape=NWbox,inner sep=2pt,minimum height=6mm,fill=white,minimum width=12mm]
\tikzstyle{large dmap trans}=[draw,doubled,shape=SWbox,inner sep=2pt,minimum height=6mm,fill=white,minimum width=12mm]
\tikzstyle{very large dmap}=[draw,doubled,shape=NEbox,inner sep=2pt,minimum height=6mm,fill=white,minimum width=19.5mm]

\tikzstyle{muxbox}=[draw,shape=rectangle,minimum height=3mm,minimum width=3mm,fill=white]
\tikzstyle{dmuxbox}=[muxbox,doubled]

\tikzstyle{box}=[draw,shape=rectangle,inner sep=2pt,minimum height=6mm,minimum width=6mm,fill=white]
\tikzstyle{dbox}=[draw,doubled,shape=rectangle,inner sep=2pt,minimum height=6mm,minimum width=6mm,fill=white]
\tikzstyle{dmap}=[draw,doubled,shape=NEbox,inner sep=2pt,minimum height=6mm,fill=white]
\tikzstyle{dmapdag}=[draw,doubled,shape=SEbox,inner sep=2pt,minimum height=6mm,fill=white]
\tikzstyle{dmapadj}=[draw,doubled,shape=SEbox,inner sep=2pt,minimum height=6mm,fill=white]
\tikzstyle{dmaptrans}=[draw,doubled,shape=SWbox,inner sep=2pt,minimum height=6mm,fill=white]
\tikzstyle{dmapconj}=[draw,doubled,shape=NWbox,inner sep=2pt,minimum height=6mm,fill=white]

\tikzstyle{ddmap}=[draw,doubled,dashed,shape=NEbox,inner sep=2pt,minimum height=6mm,fill=white]
\tikzstyle{ddmapdag}=[draw,doubled,dashed,shape=SEbox,inner sep=2pt,minimum height=6mm,fill=white]
\tikzstyle{ddmapadj}=[draw,doubled,dashed,shape=SEbox,inner sep=2pt,minimum height=6mm,fill=white]
\tikzstyle{ddmaptrans}=[draw,doubled,dashed,shape=SWbox,inner sep=2pt,minimum height=6mm,fill=white]
\tikzstyle{ddmapconj}=[draw,doubled,dashed,shape=NWbox,inner sep=2pt,minimum height=6mm,fill=white]

\boxshape{sNEbox}{0pt}{3pt}
\boxshape{sSEbox}{0pt}{-3pt}
\boxshape{sNWbox}{3pt}{0pt}
\boxshape{sSWbox}{-3pt}{0pt}
\tikzstyle{smap}=[draw,shape=sNEbox,fill=white]
\tikzstyle{smapdag}=[draw,shape=sSEbox,fill=white]
\tikzstyle{smapadj}=[draw,shape=sSEbox,fill=white]
\tikzstyle{smaptrans}=[draw,shape=sSWbox,fill=white]
\tikzstyle{smapconj}=[draw,shape=sNWbox,fill=white]

\tikzstyle{dsmap}=[draw,dashed,shape=sNEbox,fill=white]
\tikzstyle{dsmapdag}=[draw,dashed,shape=sSEbox,fill=white]
\tikzstyle{dsmaptrans}=[draw,dashed,shape=sSWbox,fill=white]
\tikzstyle{dsmapconj}=[draw,dashed,shape=sNWbox,fill=white]

\boxshape{mNEbox}{0pt}{10pt}
\boxshape{mSEbox}{0pt}{-10pt}
\boxshape{mNWbox}{10pt}{0pt}
\boxshape{mSWbox}{-10pt}{0pt}
\tikzstyle{mmap}=[draw,shape=mNEbox]
\tikzstyle{mmapdag}=[draw,shape=mSEbox]
\tikzstyle{mmaptrans}=[draw,shape=mSWbox]
\tikzstyle{mmapconj}=[draw,shape=mNWbox]

\tikzstyle{mmapgray}=[draw,fill=gray!40!white,shape=mNEbox]
\tikzstyle{smapgray}=[draw,fill=gray!40!white,shape=sNEbox]

\makeatletter
\pgfdeclareshape{cornerpoint}{
\inheritsavedanchors[from=rectangle] 
\inheritanchorborder[from=rectangle]
\inheritanchor[from=rectangle]{center}
\inheritanchor[from=rectangle]{north}
\inheritanchor[from=rectangle]{south}
\inheritanchor[from=rectangle]{west}
\inheritanchor[from=rectangle]{east}
\backgroundpath{
\southwest \pgf@xa=\pgf@x \pgf@ya=\pgf@y
\northeast \pgf@xb=\pgf@x \pgf@yb=\pgf@y

\pgfmathsetmacro{\pgf@shorten@left}{\pgfkeysvalueof{/tikz/shorten left}}
\pgfmathsetmacro{\pgf@shorten@right}{\pgfkeysvalueof{/tikz/shorten right}}

\pgfpathmoveto{\pgfpoint{0.5 * (\pgf@xa + \pgf@xb)}{\pgf@ya - 5pt}}
\pgfpathlineto{\pgfpoint{\pgf@xa - 8pt + \pgf@shorten@left}{\pgf@yb - 1.5 * \pgf@shorten@left}}
\pgfpathlineto{\pgfpoint{\pgf@xa - 8pt + \pgf@shorten@left}{\pgf@yb}}
\pgfpathlineto{\pgfpoint{\pgf@xb + 8pt - \pgf@shorten@right}{\pgf@yb}}
\pgfpathlineto{\pgfpoint{\pgf@xb + 8pt - \pgf@shorten@right}{\pgf@yb - 1.5 * \pgf@shorten@right}}
\pgfpathclose
}
}

\pgfdeclareshape{cornercopoint}{
\inheritsavedanchors[from=rectangle] 
\inheritanchorborder[from=rectangle]
\inheritanchor[from=rectangle]{center}
\inheritanchor[from=rectangle]{north}
\inheritanchor[from=rectangle]{south}
\inheritanchor[from=rectangle]{west}
\inheritanchor[from=rectangle]{east}
\backgroundpath{
\southwest \pgf@xa=\pgf@x \pgf@ya=\pgf@y
\northeast \pgf@xb=\pgf@x \pgf@yb=\pgf@y

\pgfmathsetmacro{\pgf@shorten@left}{\pgfkeysvalueof{/tikz/shorten left}}
\pgfmathsetmacro{\pgf@shorten@right}{\pgfkeysvalueof{/tikz/shorten right}}

\pgfpathmoveto{\pgfpoint{0.5 * (\pgf@xa + \pgf@xb)}{\pgf@yb + 5pt}}
\pgfpathlineto{\pgfpoint{\pgf@xa - 8pt + \pgf@shorten@left}{\pgf@ya + 1.5 * \pgf@shorten@left}}
\pgfpathlineto{\pgfpoint{\pgf@xa - 8pt + \pgf@shorten@left}{\pgf@ya}}
\pgfpathlineto{\pgfpoint{\pgf@xb + 8pt - \pgf@shorten@right}{\pgf@ya}}
\pgfpathlineto{\pgfpoint{\pgf@xb + 8pt - \pgf@shorten@right}{\pgf@ya + 1.5 * \pgf@shorten@right}}
\pgfpathclose
}
}

\makeatother

\pgfkeyssetvalue{/tikz/shorten left}{0pt}
\pgfkeyssetvalue{/tikz/shorten right}{0pt}

\tikzstyle{kpoint common}=[draw,fill=white,inner sep=1pt,minimum height=3mm]
\tikzstyle{kpoint}=[shape=cornerpoint,shorten left=5pt,kpoint common]
\tikzstyle{kpoint adjoint}=[shape=cornercopoint,shorten left=5pt,kpoint common]
\tikzstyle{kpoint conjugate}=[shape=cornerpoint,shorten right=5pt,kpoint common]
\tikzstyle{kpoint transpose}=[shape=cornercopoint,shorten right=5pt,kpoint common]
\tikzstyle{kpoint symm}=[shape=cornerpoint,shorten left=5pt,shorten right=5pt,kpoint common]

\tikzstyle{black kpoint}=[shape=cornerpoint,shorten left=5pt,kpoint common,fill=black]
\tikzstyle{black kpoint adjoint}=[shape=cornercopoint,shorten left=5pt,kpoint common,fill=black]

\tikzstyle{kpointdag}=[kpoint adjoint]
\tikzstyle{kpointadj}=[kpoint adjoint]
\tikzstyle{kpointconj}=[kpoint conjugate]
\tikzstyle{kpointtrans}=[kpoint transpose]

\tikzstyle{big kpoint}=[kpoint, minimum width=1.2 cm, minimum height=8mm, inner sep=4pt, text depth=3mm]

\tikzstyle{wide kpoint}=[kpoint, minimum width=1 cm, inner sep=2pt, text depth=-0.7 mm]
\tikzstyle{wide kpointdag}=[kpointdag, minimum width=1 cm, inner sep=2pt, text depth=0.7 mm]
\tikzstyle{wide kpointconj}=[kpointconj, minimum width=1 cm, inner sep=2pt, text depth=-0.7 mm]
\tikzstyle{wide kpointtrans}=[kpointtrans, minimum width=1 cm, inner sep=2pt, text depth=0.7 mm]

\tikzstyle{gray kpoint}=[kpoint,fill=gray!50!white]
\tikzstyle{gray kpointdag}=[kpointdag,fill=gray!50!white]
\tikzstyle{gray kpointadj}=[kpointadj,fill=gray!50!white]
\tikzstyle{gray kpointconj}=[kpointconj,fill=gray!50!white]
\tikzstyle{gray kpointtrans}=[kpointtrans,fill=gray!50!white]

\tikzstyle{gray dkpoint}=[kpoint,fill=gray!50!white,doubled]
\tikzstyle{gray dkpointdag}=[kpointdag,fill=gray!50!white,doubled]
\tikzstyle{gray dkpointadj}=[kpointadj,fill=gray!50!white,doubled]
\tikzstyle{gray dkpointconj}=[kpointconj,fill=gray!50!white,doubled]
\tikzstyle{gray dkpointtrans}=[kpointtrans,fill=gray!50!white,doubled]

\tikzstyle{white label}=[draw,fill=white,rectangle,inner sep=0.7 mm]
\tikzstyle{gray label}=[draw,fill=gray!50!white,rectangle,inner sep=0.7 mm]
\tikzstyle{black label}=[draw,fill=black,rectangle,inner sep=0.7 mm]

\tikzstyle{dkpoint}=[kpoint,doubled]
\tikzstyle{wide dkpoint}=[wide kpoint,doubled]
\tikzstyle{dkpointdag}=[kpoint adjoint,doubled]
\tikzstyle{dkcopoint}=[kpoint adjoint,doubled]
\tikzstyle{dkpointadj}=[kpoint adjoint,doubled]
\tikzstyle{dkpointconj}=[kpoint conjugate,doubled]
\tikzstyle{dkpointtrans}=[kpoint transpose,doubled]

\tikzstyle{kscalar}=[kpoint common, shape=EBox, inner xsep=-1pt, inner ysep=3pt,font=\small]
\tikzstyle{kscalarconj}=[kpoint common, shape=WBox, inner xsep=-1pt, inner ysep=3pt,font=\small]

 \tikzstyle{upground}=[circuit ee IEC,thick,ground,rotate=90,scale=2.5]
 \tikzstyle{downground}=[circuit ee IEC,thick,ground,rotate=-90,scale=2.5]
 \tikzstyle{bigground}=[regular polygon,regular polygon sides=3,draw=gray,scale=0.50,inner sep=-0.5pt,minimum width=10mm,fill=gray]

\tikzstyle{arrs}=[-latex,font=\small,auto]
\tikzstyle{arrow plain}=[arrs]
\tikzstyle{arrow dashed}=[dashed,arrs]
\tikzstyle{arrow bold}=[very thick,arrs]
\tikzstyle{arrow hide}=[draw=white!0,-]
\tikzstyle{arrow reverse}=[latex-]
\tikzstyle{cdnode}=[]

\usepackage{subcaption}

\title{Mermin Non-Locality in Abstract Process Theories}
\author{
	Stefano Gogioso
	\institute{Quantum Group \\ University of Oxford}
	\email{stefano.gogioso@cs.ox.ac.uk}
	\and
	William Zeng
	\institute{Quantum Group \\ University of Oxford}
	\email{william.zeng@cs.ox.ac.uk}
}
\def\titlerunning{Mermin Non-Locality in Abstract Process Theories}
\def\authorrunning{S. Gogioso \& W. Zeng}

\begin{document}

	\maketitle

\begin{abstract}
The study of non-locality is fundamental to the understanding of quantum mechanics. The past 50 years have seen a number of non-locality proofs, but its fundamental building blocks, and the exact role it plays in quantum protocols, has remained elusive. In this paper, we focus on a particular flavour of non-locality, generalising Mermin's argument on the GHZ state. Using strongly complementary observables, we provide necessary and sufficient conditions for Mermin non-locality in abstract process theories. We show that the existence of more phases than classical points (aka eigenstates) is not sufficient, and that the key to Mermin non-locality lies in the presence of certain algebraically non-trivial phases. This allows us to show that $\fRelCategory$, a favourite toy model for categorical quantum mechanics, is Mermin local. We show Mermin non-locality to be the key resource ensuring the device-independent security of the HBB CQ (N,N) family of Quantum Secret Sharing protocols. Finally, we challenge the unspoken assumption that the measurements involved in Mermin-type scenarios should be complementary (like the pair $X,Y$), opening the doors to a much wider class of potential experimental setups than currently employed. In short, we give conditions for Mermin non-locality tests on any number of systems, where each party has an arbitrary number of measurement choices, where each measurement has an arbitrary number of outcomes and further, that works in any abstract process theory.\\

\end{abstract}

\setcounter{tocdepth}{2} 

\section{Introduction}
	\label{section_Introduction}

	Non-locality is a fundamental property of quantum mechanics.  It impacts both foundations and application, ruling out the existence of \emph{local hidden variable theories} consistent with quantum theory \cite{bell}, and underpinning protocols like quantum key distribution \cite{Ekert1999} and quantum secret sharing \cite{HBB}. The importance of this property pushed the development of methods to characterise it both in general (e.g. the sheaf-theoretic methods of \cite{NLC-SheafSeminal}) and in specific extensions of quantum theory (e.g. the generalized probabilistic theories of \cite{barrett2007information}).
	
	We focus on a particular possibilistic class of non-locality arguments generalized from Mermin's argument \cite{mermin1990quantum} and related to the recent work on All-versus-Nothing arguments by Abramsky et al. \cite{NLC-AvN}. These experiments produce possibilistic evidence for quantum mechanical non-locality, i.e. certain measurement outcomes that can only be realized by non-local theories.  Mermin scenarios are typically described by triples $(N,M,D)$ for $N$ parties with $M$ measurement choices for each party, each having $D$ classical outcomes. Current literature generalises from the original $(3,2,2)$ scenario~\cite{mermin1990quantum} to derive non-locality proofs for the $(3,3,2)$\cite{ryu2014multisetting}, $(D+1,2,D)$\cite{zukowski-GHZ-multiport}, $(N>D, 2, D \mbox{ even})$\cite{cerf-GHZ-many}, and $(\mbox{odd }N, 2, \mbox{even }D)$\cite{lee-even-dim}. One contribution of our work is to extend the work of \cite{CQM-StrongComplementarity} to cover all $(N,M,D)$ scenarios.

In \cite{CQM-StrongComplementarity}, Coecke et al. used strong complementarity to formulate Mermin arguments within the framework of Categorical Quantum Mechanics \cite{CQM-seminal}. Not only does this approach help generalize non-locality arguments within quantum theory, but it also paved the way towards an understanding of Mermin non-locality in \emph{abstract process theories}, aka dagger symmetric monoidal categories. As a corollary, they are able to identify the difference between qubit stabilizer quantum mechanics (which is non-local) and Spekken's toy theory (which is local) in the structure of the respective phase groups \cite{CQM-StrongComplementarity,coecke2011phase}.

In Sections \ref{section_MerminMeasurements} and \ref{section_MerminLocality}, we remove implicit assumptions about phase groups and classical points from~\cite{CQM-StrongComplementarity} and use strongly complementary structures to generalise Mermin measurements to abstract process theories, defining Mermin non-locality as the existence of a Mermin measurement scenario not admitting a local hidden variable model.

\newtheorem*{theorem*}{Theorem}

In Section \ref{section_MerminLocality}, we show that strong complementarity is not sufficient to characterise Mermin non-locality. The phase group structure is shown to provide necessary algebraic conditions in abstract process theories, as summarised by our first main result:
\begin{theorem*}\hspace{-3pt}\text{\textbf{\ref{thm_MerminLocality}.}}
		Let $\CategoryC$ be a $\dagger$-SMC. If for any strongly complementary pair $(\hbox{\begin{tikzpicture} [scale=1,transform shape] 

\def\deltax{0.3} 
\def\deltay{0.5} 


\node [dot, fill=\classicalStructColour] (mult) at (0,0) {};

\end{tikzpicture}}\!,\hbox{\begin{tikzpicture} [scale=1,transform shape] 

\def\deltax{0.3} 
\def\deltay{0.5} 


\node [dot, fill=\groupStructColour] (mult) at (0,0) {};

\end{tikzpicture}}\!)$ of $\dagger$-qSCFAs the group of $\hbox{\begin{tikzpicture} [scale=1,transform shape] 

\def\deltax{0.3} 
\def\deltay{0.5} 


\node [dot, fill=\classicalStructColour] (mult) at (0,0) {};

\end{tikzpicture}}\!$-phases is a trivial algebraic extension of the subgroup of $\hbox{\begin{tikzpicture} [scale=1,transform shape] 

\def\deltax{0.3} 
\def\deltay{0.5} 


\node [dot, fill=\groupStructColour] (mult) at (0,0) {};

\end{tikzpicture}}\!$-classical points (i.e. if there exist no algebraically non-trivial $\hbox{\begin{tikzpicture} [scale=1,transform shape] 

\def\deltax{0.3} 
\def\deltay{0.5} 


\node [dot, fill=\classicalStructColour] (mult) at (0,0) {};

\end{tikzpicture}}\!$-phases), then $\CategoryC$ is Mermin local.
	\end{theorem*}

\noindent Thus $\hbox{\begin{tikzpicture} [scale=1,transform shape] 

\def\deltax{0.3} 
\def\deltay{0.5} 


\node [dot, fill=\classicalStructColour] (mult) at (0,0) {};

\end{tikzpicture}}\!$-phase groups that are trivial algebraic extensions of the respective subgroups of $\hbox{\begin{tikzpicture} [scale=1,transform shape] 

\def\deltax{0.3} 
\def\deltay{0.5} 


\node [dot, fill=\groupStructColour] (mult) at (0,0) {};

\end{tikzpicture}}\!$-classical points always lead to local hidden variable models, irregardless of whether there are enough $\hbox{\begin{tikzpicture} [scale=1,transform shape] 

\def\deltax{0.3} 
\def\deltay{0.5} 


\node [dot, fill=\groupStructColour] (mult) at (0,0) {};

\end{tikzpicture}}\!$-classical points to form a basis and/or strictly more $\hbox{\begin{tikzpicture} [scale=1,transform shape] 

\def\deltax{0.3} 
\def\deltay{0.5} 


\node [dot, fill=\classicalStructColour] (mult) at (0,0) {};

\end{tikzpicture}}\!$-phases than $\hbox{\begin{tikzpicture} [scale=1,transform shape] 

\def\deltax{0.3} 
\def\deltay{0.5} 


\node [dot, fill=\groupStructColour] (mult) at (0,0) {};

\end{tikzpicture}}\!$-classical points. Indeed, we show that the category $\fRelCategory$ of finite sets and relations is Mermin local (despite it having arbitrarily many more $\hbox{\begin{tikzpicture} [scale=1,transform shape] 

\def\deltax{0.3} 
\def\deltay{0.5} 


\node [dot, fill=\classicalStructColour] (mult) at (0,0) {};

\end{tikzpicture}}\!$-phases than $\hbox{\begin{tikzpicture} [scale=1,transform shape] 

\def\deltax{0.3} 
\def\deltay{0.5} 


\node [dot, fill=\groupStructColour] (mult) at (0,0) {};

\end{tikzpicture}}\!$-classical points), and also confirm that Spekken's toy theory is Mermin local (despite them having enough $\hbox{\begin{tikzpicture} [scale=1,transform shape] 

\def\deltax{0.3} 
\def\deltay{0.5} 


\node [dot, fill=\groupStructColour] (mult) at (0,0) {};

\end{tikzpicture}}\!$-classical points to form a basis).  Our method also gives an easy proof that qutrit stabilizer mechanic is Mermin local.

Additionally, in Section \ref{section_MerminLocality}, we show that the existence of algebraically non-trivial $\hbox{\begin{tikzpicture} [scale=1,transform shape] 

\def\deltax{0.3} 
\def\deltay{0.5} 


\node [dot, fill=\classicalStructColour] (mult) at (0,0) {};

\end{tikzpicture}}\!$-phases is sufficient, under mild additional assumptions, to formulate a non-locality argument. This leads to our second main result:
	\begin{theorem*}\hspace{-3pt}\text{\textbf{\ref{thm_MerminNonLocality}.}}
		Let $\CategoryC$ be a $\dagger$-SMC, and $(\hbox{\begin{tikzpicture} [scale=1,transform shape] 

\def\deltax{0.3} 
\def\deltay{0.5} 


\node [dot, fill=\classicalStructColour] (mult) at (0,0) {};

\end{tikzpicture}}\!,\hbox{\begin{tikzpicture} [scale=1,transform shape] 

\def\deltax{0.3} 
\def\deltay{0.5} 


\node [dot, fill=\groupStructColour] (mult) at (0,0) {};

\end{tikzpicture}}\!)$ be a strongly complementary pair of $\dagger$-qSCFAs. Suppose further that the $\hbox{\begin{tikzpicture} [scale=1,transform shape] 

\def\deltax{0.3} 
\def\deltay{0.5} 


\node [dot, fill=\groupStructColour] (mult) at (0,0) {};

\end{tikzpicture}}\!$-classical points form a basis. If the group of $\hbox{\begin{tikzpicture} [scale=1,transform shape] 

\def\deltax{0.3} 
\def\deltay{0.5} 


\node [dot, fill=\classicalStructColour] (mult) at (0,0) {};

\end{tikzpicture}}\!$-phases is a non-trivial algebraic extension of the subgroup of $\hbox{\begin{tikzpicture} [scale=1,transform shape] 

\def\deltax{0.3} 
\def\deltay{0.5} 


\node [dot, fill=\groupStructColour] (mult) at (0,0) {};

\end{tikzpicture}}\!$-classical points, then $\CategoryC$ is Mermin non-local.
	\end{theorem*}
\noindent As a consequence, we confirm that qubit stabilizer quantum mechanics is Mermin non-local.

In Section \ref{section_QSS}, we argue that our concrete characterisation as the existence of algebraically non-trivial phases can be used to see Mermin non-locality as a resource in the construction of quantum protocols. We exemplify this by showing how the security of the HBB CQ (N,N) family of Quantum Secret Sharing protocols from \cite{HBB, HBB2} directly relates to the flavour of non-locality explored in this work.

In Section \ref{section:non-compl}, we use our general framework to investigate Mermin non-locality in $\fdHilbCategory$, the usual arena of quantum mechanics. The traditional formulation of Mermin arguments relies on sets of complementary measurements, such as the $X$ ($\hbox{\begin{tikzpicture} [scale=1,transform shape] 

\def\deltax{0.3} 
\def\deltay{0.5} 


\node [dot, fill=\groupStructColour] (mult) at (0,0) {};

\end{tikzpicture}}\!$ measurement with $\hbox{\begin{tikzpicture} [scale=1,transform shape] 

\def\deltax{0.3} 
\def\deltay{0.5} 


\node [dot, fill=\classicalStructColour] (mult) at (0,0) {};

\end{tikzpicture}}\!$-phase $0$) and $Y$ ($\hbox{\begin{tikzpicture} [scale=1,transform shape] 

\def\deltax{0.3} 
\def\deltay{0.5} 


\node [dot, fill=\groupStructColour] (mult) at (0,0) {};

\end{tikzpicture}}\!$ measurement with $\hbox{\begin{tikzpicture} [scale=1,transform shape] 

\def\deltax{0.3} 
\def\deltay{0.5} 


\node [dot, fill=\classicalStructColour] (mult) at (0,0) {};

\end{tikzpicture}}\!$-phase $\frac{\pi}{2}$) measurements of the qubit in the original $(3,2,2)$ Mermin argument. We show how, even in the case of $(N,2,D)$ scenarios, many more possible measurements exist than complementary ones. This result opens a wealth of novel experimental configurations for tests of Mermin non-locality and, through results of Section \ref{section_QSS}, new configurations for quantum secret sharing protocols as well.


\section{Background}
	\label{section_Background}

This section refers the reader to background literature on the mathematical concepts of abstract process theories that we use in this work.

	Classical structures, aka special commutative $\dagger$-Frobenius algebras ($\dagger$-SCFAs), play a central role in Categorical Quantum Mechanics (CQM)\cite{CQM-seminal} as the abstract incarnation of non-degenerate observables. The operational aspect of $\dagger$-SCFAs is extensively covered in \cite{CQM-QuantumClassicalStructuralism}, where they are interpreted as models for the classical data operations of copy, deletion, and comparison. Their key connection with non-degenerate observables in quantum mechanics is provided by \cite{CQM-OrthogonalBases}, where it is proven that $\dagger$-SCFAs in $\fdHilbCategory$ canonically correspond to orthonormal bases (their unique basis of copyable, or \textit{classical}, states), and can thus be used to model a basis of eigenstates;  more generally, commutative $\dagger$-Frobenius algebras ($\dagger$-CFAs) correspond to orthogonal bases.
	
	Strongly complementary pairs of classical structures appear in \cite{CQM-StrongComplementarity,CQM-ZXCalculusSeminal} to model non-locality in terms of commutative non-degenerate observables of generalized Mermin arguments. The paper \cite{CQM-KissingerPhdthesis} shows that 
they correspond to finite abelian groups in $\fdHilbCategory$ and \cite{StefanoGogioso-RepTheoryCQM} specifies their connection to the Fourier Transform. The notion of phase groups was explicitly introduced in \cite{CQM-ZXCalculusSeminal, coecke2011phase}. Their connection to non-locality was first made in \cite{coecke2011phase}, where it was used to differentiate Spekkens toy theory from stabilizer quantum mechanics. 	Finally, the upcoming \cite{CQM-QCSnotes} and \cite{CQM-CQMnotes} provide a comprehensive reference for many structures and results used here. These, along with the survey~\cite{selinger2011survey}, are also good references for the diagrammatic notation used throughout this literature.


\section{Mermin measurements}
	\label{section_MerminMeasurements}

Unlike Bell tests, which produce outcomes with probabilities that are forbidden to local hidden variable theories, the Mermin (or GHZ) argument produces outcomes which are impossible to observe in a local hidden variable theory \cite{mermin1990quantum}. This section introduces the definitions necessary to generalise the Mermin argument to abstract process theories. We make use of the standard definitions for strongly complementary observables, phase states and phases. We often refer to quasi-special $\dagger$-Frobenius algebras as \textbf{non-degenerate observables} and use the shorthand $\dagger$-qSFA. The acronym $\dagger$-qSCFA refers to a commutative $\dagger$-qSFA. Definitions of these concepts are reproduced in Appendix \ref{app:defs}.

	\begin{definition}\label{def_basis}
		A family $(\ket{\psi_j})_j$ of states of an object $\SpaceH$ in a $\dagger$-SMC forms a (orthogonal) \textbf{basis} if the following two conditions hold:
		\begin{enumerate}
			\item[1.] $\braket{\psi_i}{\psi_j} = 0$ for $i \neq j$
			\item[2.] for any $f,g: \SpaceH \rightarrow \SpaceH'$ we have that $\forall j. \; f \ket{\psi_j} = g \ket{\psi_j}$ implies $f=g$
		\end{enumerate}
	\end{definition}
    \noindent In $\fdHilbCategory$, the objects are vector spaces and any orthogonal vector space basis clearly obeys these conditions.  The above Definition allows us to extend the appropriate notion of a basis to an arbitrary $\dagger$-SMC. Within the context of Categorical Quantum Mechanics, a $\dagger$-qSCFA $\hbox{\begin{tikzpicture} [scale=1,transform shape] 

\def\deltax{0.3} 
\def\deltay{0.5} 


\node [dot, fill=\groupStructColour] (mult) at (0,0) {};

\end{tikzpicture}}\!$ with classical points forming a basis is said to have \textbf{enough classical points}. More details on phases and classical points of observables can be found in the Appendix.

	\begin{theorem}\label{thm_PhaseGroup} 
		Let $\hbox{\begin{tikzpicture} [scale=1,transform shape] 

\def\deltax{0.3} 
\def\deltay{0.5} 


\node [dot, fill=\classicalStructColour] (mult) at (0,0) {};

\end{tikzpicture}}\!$ and $\hbox{\begin{tikzpicture} [scale=1,transform shape] 

\def\deltax{0.3} 
\def\deltay{0.5} 


\node [dot, fill=\groupStructColour] (mult) at (0,0) {};

\end{tikzpicture}}\!$ be strongly complementary $\dagger$-qSFAs in any $\dagger$-SMC. Phase states (resp. phases) of $\hbox{\begin{tikzpicture} [scale=1,transform shape] 

\def\deltax{0.3} 
\def\deltay{0.5} 


\node [dot, fill=\classicalStructColour] (mult) at (0,0) {};

\end{tikzpicture}}\!$ form group under the action of $(\ZmultSym,\ZunitSym)$. This group of phase states is denoted the {\bf phase group} P$_{\hbox{\begin{tikzpicture} [scale=1,transform shape] 

\def\deltax{0.3} 
\def\deltay{0.5} 


\node [dot, fill=\classicalStructColour] (mult) at (0,0) {};

\end{tikzpicture}}\!}$. The classical points (resp. the induced phases) of $\hbox{\begin{tikzpicture} [scale=1,transform shape] 

\def\deltax{0.3} 
\def\deltay{0.5} 


\node [dot, fill=\groupStructColour] (mult) at (0,0) {};

\end{tikzpicture}}\!$ are a subgroup K$_{\hbox{\begin{tikzpicture} [scale=1,transform shape] 

\def\deltax{0.3} 
\def\deltay{0.5} 


\node [dot, fill=\groupStructColour] (mult) at (0,0) {};

\end{tikzpicture}}\!}\subseteq $P$_{\hbox{\begin{tikzpicture} [scale=1,transform shape] 

\def\deltax{0.3} 
\def\deltay{0.5} 


\node [dot, fill=\classicalStructColour] (mult) at (0,0) {};

\end{tikzpicture}}\!}$. 
	\end{theorem}
	\begin{proof}
		Proof that phases form a group can be found in \cite{CQM-CQMnotes}. Proof that classical points form a group can be found in \cite{CQM-StrongComplementarity} (for $\dagger$-SCFAs) and \cite{StefanoGogioso-RepTheoryCQM}. Statement follows from this.
	\end{proof}
	\noindent When talking about the phase group of a $\dagger$-qSCFA is commutative, we use additive notation: given two $\hbox{\begin{tikzpicture} [scale=1,transform shape] 

\def\deltax{0.3} 
\def\deltay{0.5} 


\node [dot, fill=\classicalStructColour] (mult) at (0,0) {};

\end{tikzpicture}}\!$-phase states $\ket{\alpha}$ and $\ket{\beta}$, we denote by $\ket{\alpha+\beta}$ their addition in the phase group. From now on, we interchangeably use phase states and phases, leaving disambiguation to context.
	
\noindent The GHZ states and Mermin measurements are the main ingredients needed in our argument.  GHZ states appear in the ZX calculus fragment of our framework in~\cite{CQM-ZXCalculusSeminal} and are generalized to the definition that we use in~\cite{CQM-StrongComplementarity}.
\begin{definition}
Given a $\dagger$-qSFA $\hbox{\begin{tikzpicture} [scale=1,transform shape] 

\def\deltax{0.3} 
\def\deltay{0.5} 


\node [dot, fill=\classicalStructColour] (mult) at (0,0) {};

\end{tikzpicture}}\!$ in a $\dagger$-SMC, an \textbf{$N$-partite GHZ state} for $\hbox{\begin{tikzpicture} [scale=1,transform shape] 

\def\deltax{0.3} 
\def\deltay{0.5} 


\node [dot, fill=\classicalStructColour] (mult) at (0,0) {};

\end{tikzpicture}}\!$ is:
\begin{equation}\label{eqn_GHZstate}
			\hbox{\begin{tikzpicture}[node distance=10mm]

\node [smalldot, fill = \classicalStructColour] (0) at (0, -1) {};
\node (1) at (1, -0) {};
\node (2) at (-1, -0) {};
\node (3) at (0, -0) {$\cdot\;\cdot\;\cdot$};
                
\draw [bend right=45, looseness=1.00] (2.center) to (0);
\draw [bend right=45, looseness=1.00] (0) to (1.center);

\draw [thick, decoration={brace, amplitude=6pt},decorate] (-1.2,0.3) to (1.2,0.3);
\node at (0, 1) {n-systems};
\end{tikzpicture}}
		\end{equation}
\end{definition}

\noindent Inspired by \cite{CQM-StrongComplementarity}, we build Mermin type scenarios out of them.
	\begin{definition}\label{def_MerminMeasurements} 
		Let $\hbox{\begin{tikzpicture} [scale=1,transform shape] 

\def\deltax{0.3} 
\def\deltay{0.5} 


\node [dot, fill=\classicalStructColour] (mult) at (0,0) {};

\end{tikzpicture}}\!$ and $\hbox{\begin{tikzpicture} [scale=1,transform shape] 

\def\deltax{0.3} 
\def\deltay{0.5} 


\node [dot, fill=\groupStructColour] (mult) at (0,0) {};

\end{tikzpicture}}\!$ be a pair of strongly complementary $\dagger$-qSFAs in a $\dagger$-SMC. An $N$-partite \textbf{Mermin measurement} is obtained by applying $N$ $\hbox{\begin{tikzpicture} [scale=1,transform shape] 

\def\deltax{0.3} 
\def\deltay{0.5} 


\node [dot, fill=\classicalStructColour] (mult) at (0,0) {};

\end{tikzpicture}}\!$-phases $\alpha_1,...,\alpha_N$ to the $N$ components of an $N$-partite GHZ state, and then measuring each component in the $\hbox{\begin{tikzpicture} [scale=1,transform shape] 

\def\deltax{0.3} 
\def\deltay{0.5} 


\node [dot, fill=\groupStructColour] (mult) at (0,0) {};

\end{tikzpicture}}\!$ structure:
		\begin{equation}\label{eqn_MerminMeasurementGHZstate}
		    \hbox{\begin{tikzpicture}[node distance=10mm, yscale=1.6, xscale=1.8]

                \node [smalldot, fill = \classicalStructColour] (0) at (-0.5, -2) {};
                \node [smalldot, fill = \classicalStructColour] (1) at (0.5, -2) {};
                \node [style=dot, fill = \classicalStructColour, inner sep=0pt] (2) at (-2, -0.5) {$-\alpha_1$};
                \node [style=dot, fill = \classicalStructColour, inner sep=3pt] (3) at (-1, -0.5) {$\alpha_1$};
                \node [style=dot, fill = \classicalStructColour, inner sep=0pt] (4) at (1, -0.5) {$-\alpha_N$};
                \node [style=dot, fill = \classicalStructColour, inner sep=3pt] (5) at (2, -0.5) {$\alpha_N$};
                \node [smalldot, fill = \groupStructColour] (6) at (1.5, 0.5) {};
                \node [smalldot, fill = \groupStructColour] (7) at (-1.5, 0.5) {};
                \node (8) at (1.5, 1.25) {};
                \node (9) at (-1.5, 1.25) {};
                \node (10) at (0, 0.25) {$\cdot\;\cdot\;\cdot$};

                \draw [->-=.5] (2) to (0);
                \draw [->-=.5] (1) to (3);
                \draw [->-=.5] (4) to (0);
                \draw [->-=.5] (1) to (5);
                \draw [->-=.5, bend right=45, looseness=1.00] (5) to (6);
                \draw [->-=.5, bend right=45, looseness=1.00] (6) to (4);
                \draw [->-=.5, bend right=45, looseness=1.00] (3) to (7);
                \draw [->-=.5, bend right=45, looseness=1.00] (7) to (2);
                \draw [->-=.5] (7) to (9.center);
                \draw [->-=.5] (6) to (8.center);

\end{tikzpicture}}
		\end{equation}
		We further require that $\sum_i \alpha_i$, where the sum is taken in the group of phases, be a $\hbox{\begin{tikzpicture} [scale=1,transform shape] 

\def\deltax{0.3} 
\def\deltay{0.5} 


\node [dot, fill=\groupStructColour] (mult) at (0,0) {};

\end{tikzpicture}}\!$-classical point.
	\end{definition}

	\begin{lemma}\label{thm_MerminMeasurementGHZstate}
		The Mermin measurement shown in Equation \ref{eqn_MerminMeasurementGHZstate} is equivalent to the following state:
	\begin{equation}\label{eqn_MerminMeasurementGHZstateSimplified}
	    \hbox{\begin{tikzpicture}[node distance=8mm, xscale=1.8, yscale=0.9]

                \node [style=dot, fill = \classicalStructColour, inner sep=1pt] (0) at (-0.75, -1.95) {\small $-\sum\alpha_i$};
                \node [style=dot, fill = \classicalStructColour, inner sep=1pt] (1) at (0.75, -1.95) {\small $+\sum\alpha_i$};
                \node [smalldot, fill = \groupStructColour] (2) at (0, 0) {};
                \node [smalldot, fill = \classicalStructColour] (3) at (0, 1) {};
                \node (4) at (-1.5,2.5) {};
                \node (5) at (-1, 2.5) {};
                \node (6) at (1.5, 2.5) {};
                \node (7) at (0, 2.5) {$\cdot\cdot\cdot$};

        \begin{pgfonlayer}{background}
                \draw [->-=.5, bend left, looseness=1.00] (3) to (4.center);
                \draw [->-=.5, bend left=15, looseness=0.75] (3) to (5.center);
                \draw [->-=.5, bend right=45, looseness=1.00] (3) to (6.center);
                \draw [->-=.5, bend right=45, looseness=0.75] (1.center) to (2);
                \draw [->-=.5, bend right=45, looseness=0.75] (2) to (0.center);
                \draw [->-=.5] (2) to (3);
        \end{pgfonlayer}

\end{tikzpicture}}
	\end{equation}		
	\end{lemma}
	\begin{proof}
		Pushing the phases down through the $\hbox{\begin{tikzpicture} [scale=1,transform shape] 

\def\deltax{0.3} 
\def\deltay{0.5} 


\node [dot, fill=\classicalStructColour] (mult) at (0,0) {};

\end{tikzpicture}}\!$ nodes and using strong complementarity. See \cite{CQM-StrongComplementarity}.
	\end{proof}

\noindent While this defines a single Mermin experiment, the full non-locality argument requires the joint outcomes of several Mermin measurements.
	\begin{definition} 
		Let $\hbox{\begin{tikzpicture} [scale=1,transform shape] 

\def\deltax{0.3} 
\def\deltay{0.5} 


\node [dot, fill=\classicalStructColour] (mult) at (0,0) {};

\end{tikzpicture}}\!$ and $\hbox{\begin{tikzpicture} [scale=1,transform shape] 

\def\deltax{0.3} 
\def\deltay{0.5} 


\node [dot, fill=\groupStructColour] (mult) at (0,0) {};

\end{tikzpicture}}\!$ be strongly complementary $\dagger$-qSCFAs on a space $\SpaceH$ in a $\dagger$-SMC. An $N$-partite \textbf{Mermin measurement scenario} (for $\hbox{\begin{tikzpicture} [scale=1,transform shape] 

\def\deltax{0.3} 
\def\deltay{0.5} 


\node [dot, fill=\classicalStructColour] (mult) at (0,0) {};

\end{tikzpicture}}\!$ and $\hbox{\begin{tikzpicture} [scale=1,transform shape] 

\def\deltax{0.3} 
\def\deltay{0.5} 


\node [dot, fill=\groupStructColour] (mult) at (0,0) {};

\end{tikzpicture}}\!$) is any non-empty, finite collection of Mermin measurements $\underline{\alpha}^s = (\alpha_1^s,...,\alpha_N^s)_{s=1,...,S}$ of the $N$-partite GHZ state in the form of Equation \ref{thm_MerminMeasurementGHZstate}.
	\end{definition}
\noindent In the category $\fdHilbCategory$ of finite-dimensional Hilbert spaces, an $N$-partite Mermin measurement scenario where $a_1,...,a_M$ are the distinct $\hbox{\begin{tikzpicture} [scale=1,transform shape] 

\def\deltax{0.3} 
\def\deltay{0.5} 


\node [dot, fill=\classicalStructColour] (mult) at (0,0) {};

\end{tikzpicture}}\!$-phases appearing in the scenario and $\SpaceH$ is $D$-dimensional is exactly the usual $(N,M,D)$ Mermin scenario.  This correspondence is clarified in Section \ref{section_MerminLocality}, where we derive our generalized Mermin non-locality argument.


\section{Mermin locality and non-locality}
	\label{section_MerminLocality}
	
	The last definitions we need for our main results, Theorems \ref{thm_MerminNonLocality} and \ref{thm_MerminLocality}, are those of local hidden variable models (following the construction of \cite{CQM-StrongComplementarity}) and non-trivial algebraic extensions.

	\begin{definition}\label{def_LHVMerminMeasurement}
		Let $\hbox{\begin{tikzpicture} [scale=1,transform shape] 

\def\deltax{0.3} 
\def\deltay{0.5} 


\node [dot, fill=\classicalStructColour] (mult) at (0,0) {};

\end{tikzpicture}}\!$ and $\hbox{\begin{tikzpicture} [scale=1,transform shape] 

\def\deltax{0.3} 
\def\deltay{0.5} 


\node [dot, fill=\groupStructColour] (mult) at (0,0) {};

\end{tikzpicture}}\!$ be strongly complementary $\dagger$-qSCFAs on some system $\SpaceH$. Consider an $N$-partite Mermin measurement scenario $(\underline{\alpha}^s)_{s=1,...,S}$, and let $a_1,...,a_M$ be the distinct $\hbox{\begin{tikzpicture} [scale=1,transform shape] 

\def\deltax{0.3} 
\def\deltay{0.5} 


\node [dot, fill=\classicalStructColour] (mult) at (0,0) {};

\end{tikzpicture}}\!$-phases appearing in it. The \textbf{local map} for the scenario is the map $\SpaceH^{\tensor (M \cdot N)} \rightarrow \SpaceH^{\tensor (N \cdot S)}$ defined as follows:
		\begin{enumerate}
			\item[a.] we group the input wires in $N$ groups of $M$ wires: we say that the $r$-th wire of $i$-th group is the $a_r$ \textbf{input wire for system } $i$
			\item[b.] we group the output wires in $S$ groups of $N$ wires: we say that the $j$-th wire of $r$-th group is the $j$-th \textbf{output wire for measurement } $s$
			\item[c.] each input wire is connected to a $\hbox{\begin{tikzpicture} [scale=1,transform shape] 

\def\deltax{0.3} 
\def\deltay{0.5} 


\node [dot, fill=\groupStructColour] (mult) at (0,0) {};

\end{tikzpicture}}\!$ node
			\item[d.] for all $r,i,j$ and $s$, the $\hbox{\begin{tikzpicture} [scale=1,transform shape] 

\def\deltax{0.3} 
\def\deltay{0.5} 


\node [dot, fill=\groupStructColour] (mult) at (0,0) {};

\end{tikzpicture}}\!$ node of each $a_r$ input wire for system $i$ is connected to the $j$-th output wire for measurement $s$ if and only if $i=j$ and $\alpha_j^s = a_r$
		\end{enumerate}
		The following diagram details the procedure:
		\begin{equation}\label{eqn_LocalMap}
			\hbox{\begin{tikzpicture}[node distance = 10mm]



\node (a1sys1dot) [smalldot, fill = \Xcolour] {};
\node (a1sys1dotl) [above of = a1sys1dot, xshift = -4mm,yshift = -3mm] {};
\node (a1sys1dotr) [above of = a1sys1dot, xshift = +4mm,yshift = -3mm] {};
\node (a1sys1dotellipsis) [above of = a1sys1dot, yshift = -2mm,yshift = -3mm] {$...$};

\node (a1sys1) [below of = a1sys1dot,yshift = +3mm] {$a_1$};

\node (sys1dotellipsis) [right of = a1sys1dot,xshift = -3mm] {$\cdot \cdot \cdot$};
\node (sys1label) [below of = sys1dotellipsis, yshift = -2mm] {System $1$};

\node (aMsys1dot) [smalldot, right of = sys1dotellipsis,xshift = -3mm, fill = \Xcolour] {};
\node (aMsys1dotl) [above of = aMsys1dot, xshift = -4mm,yshift = -3mm] {};
\node (aMsys1dotr) [above of = aMsys1dot, xshift = +4mm,yshift = -3mm] {};
\node (aMsys1dotellipsis) [above of = aMsys1dot, yshift = -2mm,yshift = -3mm] {$...$};

\node (aMsys1) [below of = aMsys1dot,yshift = +3mm] {$a_M$};

\node (sys1sysiellipsis) [right of = aMsys1dot, yshift = -3mm, xshift = -1mm] {};


\node (a1sysidot) [smalldot, fill = \Xcolour, right of = sys1sysiellipsis, yshift = +3mm, xshift = -1mm] {};
\node (a1sysidotl) [above of = a1sysidot, xshift = -4mm,yshift = -3mm] {};
\node (a1sysidotr) [above of = a1sysidot, xshift = +4mm,yshift = -3mm] {};
\node (a1sysidotellipsis) [above of = a1sysidot, yshift = -2mm,yshift = -3mm] {$...$};

\node (a1sysi) [below of = a1sysidot,yshift = +3mm] {$a_1$};

\node (sysidotellipsis) [right of = a1sysidot,xshift = -3mm] {$\cdot \cdot \cdot$};

\node (arsysidot) [smalldot, fill = \Xcolour, right of = sysidotellipsis,xshift = -3mm] {};
\node (arsysidotl) [above of = arsysidot, xshift = -4mm,yshift = -3mm] {};
\node (arsysidotc) [above of = arsysidot, xshift = +4mm,yshift = -3mm] {};
\node (arsysidotr) [above of = arsysidot, xshift = +8mm,yshift = -3mm] {};
\node (arsysidotellipsis) [above of = arsysidot, yshift = -2mm,yshift = -3mm] {$...$};
\node (arsysidotellipsis) [above of = arsysidot, xshift = +6mm,yshift = -2mm,yshift = -3mm] {$...$};

\node (arsysi) [below of = arsysidot,yshift = +3mm] {$a_r$};

\node (sysilabel) [below of = arsysidot, yshift = -2mm] {System $i$};
\node (sysidotellipsis2) [right of = arsysidot,xshift = -3mm] {$\cdot \cdot \cdot$};

\node (aMsysidot) [smalldot, fill = \Xcolour, right of = sysidotellipsis2,xshift = -3mm] {};
\node (aMsysidotl) [above of = aMsysidot, xshift = -4mm,yshift = -3mm] {};
\node (aMsysidotr) [above of = aMsysidot, xshift = +4mm,yshift = -3mm] {};
\node (aMsysidotellipsis) [above of = aMsysidot, yshift = -2mm,yshift = -3mm] {$...$};

\node (aMsysi) [below of = aMsysidot,yshift = +3mm] {$a_M$};

\node (sysisysNellipsis) [right of = aMsysidot, yshift = -3mm, xshift = -1mm] {};


\node (a1sysNdot) [smalldot, fill = \Xcolour, right of = sysisysNellipsis, yshift = +3mm, xshift = -1mm] {};
\node (a1sysNdotl) [above of = a1sysNdot, xshift = -4mm,yshift = -3mm] {};
\node (a1sysNdotr) [above of = a1sysNdot, xshift = +4mm,yshift = -3mm] {};
\node (a1sysNdotellipsis) [above of = a1sysNdot, yshift = -2mm,yshift = -3mm] {$...$};

\node (a1sysN) [below of = a1sysNdot,yshift = +3mm] {$a_1$};

\node (sysNdotellipsis) [right of = a1sysNdot,xshift = -3mm] {$\cdot \cdot \cdot$};
\node (sysNlabel) [below of = sysNdotellipsis, yshift = -2mm] {System $N$};

\node (aMsysNdot) [smalldot, fill = \Xcolour, right of = sysNdotellipsis,xshift = -3mm] {};
\node (aMsysNdotl) [above of = aMsysNdot, xshift = -4mm,yshift = -3mm] {};
\node (aMsysNdotr) [above of = aMsysNdot, xshift = +4mm,yshift = -3mm] {};
\node (aMsysNdotellipsis) [above of = aMsysNdot, yshift = -2mm,yshift = -3mm] {$...$};

\node (aMsysN) [below of = aMsysNdot,yshift = +3mm] {$a_M$};

\begin{pgfonlayer}{background}
\draw[->-=.5,out=90,in=270] (a1sys1) to (a1sys1dot);
\draw[->-=.5,out=90,in=270] (a1sys1dot) to (a1sys1dotl);
\draw[->-=.5,out=90,in=270] (a1sys1dot) to (a1sys1dotr);

\draw[->-=.5,out=90,in=270] (aMsys1) to (aMsys1dot);
\draw[->-=.5,out=90,in=270] (aMsys1dot) to (aMsys1dotl);
\draw[->-=.5,out=90,in=270] (aMsys1dot) to (aMsys1dotr);

\draw[->-=.5,out=90,in=270] (a1sysi) to (a1sysidot);
\draw[->-=.5,out=90,in=270] (a1sysidot) to (a1sysidotl);
\draw[->-=.5,out=90,in=270] (a1sysidot) to (a1sysidotr);

\draw[->-=.5,out=90,in=270] (arsysi) to (arsysidot);
\draw[->-=.5,out=90,in=270] (arsysidot) to (arsysidotl);
\draw[-,out=90,in=270] (arsysidot) to (arsysidotc);
\draw[->-=.5,out=90,in=270] (arsysidot) to (arsysidotr);

\draw[->-=.5,out=90,in=270] (aMsysi) to (aMsysidot);
\draw[->-=.5,out=90,in=270] (aMsysidot) to (aMsysidotl);
\draw[->-=.5,out=90,in=270] (aMsysidot) to (aMsysidotr);

\draw[->-=.5,out=90,in=270] (a1sysN) to (a1sysNdot);
\draw[->-=.5,out=90,in=270] (a1sysNdot) to (a1sysNdotl);
\draw[->-=.5,out=90,in=270] (a1sysNdot) to (a1sysNdotr);

\draw[->-=.5,out=90,in=270] (aMsysN) to (aMsysNdot);
\draw[->-=.5,out=90,in=270] (aMsysNdot) to (aMsysNdotl);
\draw[->-=.5,out=90,in=270] (aMsysNdot) to (aMsysNdotr);

\end{pgfonlayer}



\node (alfa1dotellipsis) [above of = sys1dotellipsis, yshift = 20mm] {$...$};
\node (alfa1label) [above of = alfa1dotellipsis, yshift = 0mm] {Measurement $1$};

\node (alfa1sys1anchor) [left of = alfa1dotellipsis, yshift = -5mm, xshift = 5mm] {};
\node (alfa1sys1) [above of = alfa1sys1anchor,yshift = -2mm] {$\alpha^1_1$};

\node (alfa1sysNanchor) [right of = alfa1dotellipsis,yshift = -5mm, xshift = -5mm] {};
\node (alfa1sysN) [above of = alfa1sysNanchor,yshift = -2mm] {$\alpha^1_N$};


\node (alfasanchor) [above of = arsysidot, yshift = 20mm] {};
\node (alfaslabel) [above of = alfasanchor, yshift = 0mm] {Measurement $s$};
\node (alfassysianchor) [below of = alfasanchor, yshift = 5mm] {};
\node (alfassysi) [above of = alfassysianchor,yshift = -2mm] {$\alpha^s_j$};

\node (alfasdotellipsisl) [left of = alfasanchor, xshift = 5mm] {$...$};
\node (alfasdotellipsisr) [right of = alfasanchor, xshift = -5mm] {$...$};

\node (alfassys1anchor) [left of = alfasdotellipsisl, yshift = -5mm, xshift = 5mm] {};
\node (alfassys1) [above of = alfassys1anchor,yshift = -2mm] {$\alpha^s_1$};

\node (alfassysNanchor) [right of = alfasdotellipsisr,yshift = -5mm, xshift = -5mm] {};
\node (alfassysN) [above of = alfassysNanchor,yshift = -2mm] {$\alpha^s_N$};


\node (alfaSdotellipsis) [above of = sysNdotellipsis, yshift = 20mm] {$...$};
\node (alfaSlabel) [above of = alfaSdotellipsis, yshift = 0mm] {Measurement $S$};

\node (alfaSsys1anchor) [left of = alfaSdotellipsis, yshift = -5mm, xshift = 5mm] {};
\node (alfaSsys1) [above of = alfaSsys1anchor,yshift = -2mm] {$\alpha^S_1$};

\node (alfaSsysNanchor) [right of = alfaSdotellipsis,yshift = -5mm, xshift = -5mm] {};
\node (alfaSsysN) [above of = alfaSsysNanchor,yshift = -2mm] {$\alpha^S_N$};

\begin{pgfonlayer}{background}
\draw[->-=.5,out=90,in=270] (alfa1sys1anchor) to (alfa1sys1);
\draw[->-=.5,out=90,in=270] (alfa1sysNanchor) to (alfa1sysN);

\draw[->-=.5,out=90,in=270] (alfassys1anchor) to (alfassys1);
\draw[->-=.5,out=90,in=270] (alfassysianchor) to (alfassysi);
\draw[->-=.5,out=90,in=270] (alfassysNanchor) to (alfassysN);

\draw[->-=.5,out=90,in=270] (alfaSsys1anchor) to (alfaSsys1);
\draw[->-=.5,out=90,in=270] (alfaSsysNanchor) to (alfaSsysN);
\end{pgfonlayer}


\node (label) [above of = arsysidotc, xshift = -2mm,yshift = -2mm] {Connected iff $i=j$ and $a_r = \alpha^s_j$};

\node (localmaplabel) [below of = alfa1sys1anchor, yshift = 7mm] {\textbf{Local Map}};

\begin{pgfonlayer}{background}
\node (box) [box, above of = label, yshift = -13mm, xshift = -3mm, minimum width = 110mm, minimum height = 28mm, fill = none] {};
\end{pgfonlayer}

\begin{pgfonlayer}{background}
\draw[->-=.5,out=90,in=270] (arsysidotc.270) to (label);
\draw[-,out=90,in=270] (label) to (alfassysianchor.90);
\end{pgfonlayer}

\end{tikzpicture}}
		\end{equation}		

\noindent A \textbf{local hidden variable model} for an $N$-partite Mermin measurement scenario is a state $\Lambda$ of $\SpaceH^{\tensor (N \cdot S)}$, obtained by applying the local map for the scenario to some state $\Lambda'$ of $\SpaceH^{\tensor (M \cdot N)}$. We further require that for each $s=1,...,S$, the Mermin measurement $\underline{\alpha}^s$ is the same as the state obtained from $\Lambda$ by composing an $\XcounitSym$ with each output wires of each measurement $t$ with $t \neq s$: 
		\begin{equation}\label{eqn_MerminLHVCondition}
			\hbox{\begin{tikzpicture}[node distance = 9.5mm]

\node (equals) {$=$};


\node (LHSanchor) [left of = equals, xshift = -30mm]{$\cdot \cdot \cdot$};

\node (alphajNstar) [dot,inner sep = 0.6mm,right of = LHSanchor, fill = \Zcolour] {$-\alpha^s_N$};
\node (alphajN) [dot,inner sep = 0.6mm,right of = alphajNstar, xshift=3mm, fill = \Zcolour] {$+\alpha^s_N$};

\node (alphaj1) [dot,inner sep = 1mm,left of = LHSanchor, fill = \Zcolour] {$+\alpha^s_1$};
\node (alphaj1star) [dot,inner sep = 1mm,left of = alphaj1, xshift=-3mm, fill = \Zcolour] {$-\alpha^s_1$};

\node (reddotl) [smalldot, above of = LHSanchor, yshift = 5mm, xshift = -15mm,fill = \Xcolour] {};
\node (reddotr) [smalldot, above of = LHSanchor, yshift = 5mm, xshift = +15mm,fill = \Xcolour] {};

\node (outl) [above of = reddotl, xshift = +9mm]{$\alpha^s_1$};
\node (outdots) [above of = LHSanchor, yshift = 15mm]{$...$};
\node (outr) [above of = reddotr, xshift = -9mm]{$\alpha^s_N$};

\node (greendotl) [smalldot, below of = LHSanchor, yshift = -5mm, xshift = -7mm,fill = \Zcolour] {};
\node (greendotr) [smalldot, below of = LHSanchor, yshift = -5mm, xshift = +7mm,fill = \Zcolour] {};

\begin{pgfonlayer}{background}
\draw[->-=.5,out=90,in=315] (alphaj1) to (reddotl);
\draw[-<-=.5,out=90,in=225] (alphaj1star) to (reddotl);

\draw[->-=.5,out=90,in=315] (alphajN) to (reddotr);
\draw[-<-=.5,out=90,in=225] (alphajNstar) to (reddotr);

\draw[-<-=.5,out=135,in=270] (greendotl) to (alphaj1star);
\draw[-<-=.5,out=45,in=270] (greendotl) to (alphajNstar);

\draw[->-=.5,out=135,in=270] (greendotr) to (alphaj1);
\draw[->-=.5,out=45,in=270] (greendotr) to (alphajN);

\draw[->-=.5,out=90,in=270] (reddotl) to (outl);
\draw[->-=.5,out=90,in=270] (reddotr) to (outr);

\end{pgfonlayer}


\node (RHSanchor) [right of = equals, xshift = 15mm, yshift = -4mm] {};

\node (lambda) [wide kpoint, below of = RHSanchor, minimum width = 30mm, minimum height = 9mm] {$\Lambda'$};

\node (a1sys1) [above of = lambda, xshift = -16mm] {};
\node (ellipsissys1) [above of = lambda, xshift = -11mm] {$\cdot\cdot\cdot$};
\node (aMsys1) [above of = lambda, xshift = -7mm] {};

\node (a1sysN) [above of = lambda, xshift = +16mm] {};
\node (ellipsissysN) [above of = lambda, xshift = +11mm] {$\cdot\cdot\cdot$};
\node (aMsysN) [above of = lambda, xshift = +7mm] {};

\node (localmap) [box, above of = RHSanchor, minimum height = 17mm, minimum width = 36mm] {
\textbf{Local Map}};

\node (RHSanchorhigh) [above of = RHSanchor, yshift = 15mm] {};

\node (RHSanchordots) [above of = RHSanchorhigh, yshift = -6mm] {$...$};

\node (alfa11anchor) [below of = RHSanchorhigh, xshift = -16mm, yshift = +4mm] {};
\node (alfa11) [smalldot, above of = alfa11anchor, yshift = -4mm, fill = \Xcolour] {};

\node (meas1ellipsis) [below of = RHSanchorhigh, xshift = -13mm, yshift = +10mm] {$...$};

\node (alfa1nanchor) [below of = RHSanchorhigh, xshift = -10mm, yshift = +4mm] {};
\node (alfa1n) [smalldot, above of = alfa1nanchor, yshift = -4mm, fill = \Xcolour] {};

\node (alfaj1anchor) [below of = RHSanchorhigh, xshift = -3mm, yshift = +4mm] {};
\node (alfaj1) [above of = alfaj1anchor, yshift = +0mm,xshift = -3mm] {$\alpha^s_1$};

\node (alfajnanchor) [below of = RHSanchorhigh, xshift = +3mm, yshift = +4mm] {};
\node (alfajn) [above of = alfajnanchor, yshift = +0mm,xshift = +3mm] {$\alpha^s_N$};

\node (alfaS1anchor) [below of = RHSanchorhigh, xshift = +16mm, yshift = +4mm] {};
\node (alfaS1) [smalldot, above of = alfaS1anchor, yshift = -4mm, fill = \Xcolour] {};

\node (measSellipsis) [below of = RHSanchorhigh, xshift = +13mm, yshift = +10mm] {$...$};

\node (alfaSnanchor) [below of = RHSanchorhigh, xshift = +10mm, yshift = +4mm] {};
\node (alfaSn) [smalldot, above of = alfaSnanchor, yshift = -4mm, fill = \Xcolour] {};

\begin{pgfonlayer}{background}

\draw[->-=.5,out=90,in=270] (lambda.150) to (a1sys1.90);
\draw[->-=.5,out=90,in=270] (lambda.135) to (aMsys1.90);

\draw[->-=.5,out=90,in=270] (lambda.45) to (a1sysN.90);
\draw[->-=.5,out=90,in=270] (lambda.60) to (aMsysN.90);

\draw[->-=.5,out=90,in=270] (alfa11anchor.270) to (alfa11);
\draw[->-=.5,out=90,in=270] (alfa1nanchor.270) to (alfa1n);
\draw[->-=.5,out=90,in=270] (alfaj1anchor.270) to (alfaj1);
\draw[->-=.5,out=90,in=270] (alfajnanchor.270) to (alfajn);
\draw[->-=.5,out=90,in=270] (alfaS1anchor.270) to (alfaS1);
\draw[->-=.5,out=90,in=270] (alfaSnanchor.270) to (alfaSn);

\end{pgfonlayer}

\end{tikzpicture}}
		\end{equation}	
	\end{definition}

\noindent The definition of local hidden variables finally allows us to formulate our generalised notion of Mermin non-locality. 

	\begin{definition}
		We say a $\dagger$-SMC $\CategoryC$ is \textbf{Mermin non-local} if there exists a Mermin scenario for some strongly complementary pair $(\hbox{\begin{tikzpicture} [scale=1,transform shape] 

\def\deltax{0.3} 
\def\deltay{0.5} 


\node [dot, fill=\classicalStructColour] (mult) at (0,0) {};

\end{tikzpicture}}\!,\hbox{\begin{tikzpicture} [scale=1,transform shape] 

\def\deltax{0.3} 
\def\deltay{0.5} 


\node [dot, fill=\groupStructColour] (mult) at (0,0) {};

\end{tikzpicture}}\!)$ of $\dagger$-qSCFAs which has no local hidden variable model. If for all strongly complementary pairs no such measurement exists, then we say that $\CategoryC$ is \textbf{Mermin local}.
	\end{definition}

\noindent Mermin non-locality will shortly be shown to be equivalent to the following algebraic property of the group of $\hbox{\begin{tikzpicture} [scale=1,transform shape] 

\def\deltax{0.3} 
\def\deltay{0.5} 


\node [dot, fill=\classicalStructColour] (mult) at (0,0) {};

\end{tikzpicture}}\!$-phases. The following examples will be used later on to investigate some abstract process theories of interest.

	\begin{definition} \label{def_algExt}
		Let $(G,+,0)$ be an abelian group and $(H,+,0)$ be a subgroup. We say that $G$ is a \textbf{non-trivial algebraic extension} of $H$ if there exists a finite system of equations $(\sum_{j=1}^{l} n^p_j \cdot x_j = h^p)_p$, with $h^p \in H$ and $n^p_j \in \integers$, which has solutions in $G$ but not in $H$. Otherwise, we say $G$ is a \textbf{trivial algebraic extension} of $H$. 
	\end{definition}
If $G = P_{\; \hbox{\begin{tikzpicture} [scale=1,transform shape] 

\def\deltax{0.3} 
\def\deltay{0.5} 


\node [dot, fill=\classicalStructColour] (mult) at (0,0) {};

\end{tikzpicture}}\!}$ is a non-trivial algebraic extension of $H = K_{\; \hbox{\begin{tikzpicture} [scale=1,transform shape] 

\def\deltax{0.3} 
\def\deltay{0.5} 


\node [dot, fill=\groupStructColour] (mult) at (0,0) {};

\end{tikzpicture}}\!}$, then the $\hbox{\begin{tikzpicture} [scale=1,transform shape] 

\def\deltax{0.3} 
\def\deltay{0.5} 


\node [dot, fill=\classicalStructColour] (mult) at (0,0) {};

\end{tikzpicture}}\!$-phases involved in any solution $x_j := \alpha_j$ to a system unsolvable in $K_{\;\hbox{\begin{tikzpicture} [scale=1,transform shape] 

\def\deltax{0.3} 
\def\deltay{0.5} 


\node [dot, fill=\groupStructColour] (mult) at (0,0) {};

\end{tikzpicture}}\!}$ will be called \textbf{algebraically non-trivial phases}.

	\begin{example}\label{example_PhaseGroupZX}
		Let $G = \{0,\pi/2,\pi,-\pi/2\} < \reals / 2\pi\integers$ and $H = \{0,\pi\} < G$. Then $G$ is a non-trivial algebraic extension of $H$, because the single equation $2x = \pi$ has no solution in $H$ but has solution(s) $\pm \pi/2$ in $G$. It is in fact this example that yields the original argument in $\fdHilbCategory$ from \cite{CQM-StrongComplementarity}.
	\end{example}

	\begin{lemma}\label{thm_EqnToSystem}
		Let $(G,+,0)$ be an abelian group and $(H,+,0)$ be a subgroup. Suppose that there is a function $\Phi: G \rightarrow H$ such that for any equation $\sum_{j=1}^{l} n_j \cdot x_j = h$ with $h \in H$ and $n_j \in \integers$, if $x_j := g_j$ is a solution in $G$, $x_j := \Phi(g_j)$ is also a solution (in $H$). Then $G$ is a trivial algebraic extension of $H$.
	\end{lemma}
	\begin{proof}
		Consider a system with solution $x_j := g_j$ in $G$. Then $x_j := \Phi(g_j)$ solves each individual equation in $H$, and thus also the system. 
	\end{proof}
	
	\begin{example}\label{example_PhaseGroupRel}
		Let $(K,+,0)$ be any finite abelian group, and $G = K \times K'$ for some finite non-trivial abelian group $(K',+,0)$. Let $H<G$ be the subgroup $K \times \{0\}$. If $h = (k,0) \in H$, then any equation $\sum_{j=1}^{N} n_j \cdot x_j = h$ is equivalent to the following pair of equations, where $\pi_{K}$ and $\pi_{K'}$ are the quotient projections onto $K \isom G/K'$ and $K' \isom G/K$ respectively: 
\begin{enumerate}
	\item[a.] $\sum_{j=1}^{N} n_j \cdot \pi_{K}x_j = k$ in $K$
	\item[b.] $\sum_{j=1}^{N} n_j \cdot \pi_{K'}x_j = 0$ in $K'$
\end{enumerate}	
If $x_j := g_j = (\pi_K g_j, \pi_{K'} g_j)$ is a solution in $G$, then $x_j := (\pi_K g_j,0)$ is a solution in $H$. Define $\Phi$ to be the map $g_j:G \mapsto (\pi_K g_j,0) \in H$ and use Lemma \ref{thm_EqnToSystem} to conclude that $G$ is a trivial algebraic extension of $H$. 
	\end{example}

We are now able to introduce our first main result:
	\begin{theorem}[Mermin Non-Locality]
	\label{thm_MerminNonLocality}
		Let $\CategoryC$ be a $\dagger$-SMC, and $(\hbox{\begin{tikzpicture} [scale=1,transform shape] 

\def\deltax{0.3} 
\def\deltay{0.5} 


\node [dot, fill=\classicalStructColour] (mult) at (0,0) {};

\end{tikzpicture}}\!,\hbox{\begin{tikzpicture} [scale=1,transform shape] 

\def\deltax{0.3} 
\def\deltay{0.5} 


\node [dot, fill=\groupStructColour] (mult) at (0,0) {};

\end{tikzpicture}}\!)$ be a strongly complementary pair of $\dagger$-qSCFAs. Suppose further that the $\hbox{\begin{tikzpicture} [scale=1,transform shape] 

\def\deltax{0.3} 
\def\deltay{0.5} 


\node [dot, fill=\groupStructColour] (mult) at (0,0) {};

\end{tikzpicture}}\!$-classical points form a basis. If the group of $\hbox{\begin{tikzpicture} [scale=1,transform shape] 

\def\deltax{0.3} 
\def\deltay{0.5} 


\node [dot, fill=\classicalStructColour] (mult) at (0,0) {};

\end{tikzpicture}}\!$-phases is a non-trivial algebraic extension of the subgroup of $\hbox{\begin{tikzpicture} [scale=1,transform shape] 

\def\deltax{0.3} 
\def\deltay{0.5} 


\node [dot, fill=\groupStructColour] (mult) at (0,0) {};

\end{tikzpicture}}\!$-classical points, then $\CategoryC$ is Mermin non-local.
	\end{theorem}
	\begin{proof} 
		For clarity, we present a proof where the system of equations that defines the phase group as a non-trivial algebraic extension is composed of a single equation. The construction for general systems of $l$ equations consists of $l$ copies of the construction we explicitly give. 
		
		Let $a_1, ..., a_M$ be $\hbox{\begin{tikzpicture} [scale=1,transform shape] 

\def\deltax{0.3} 
\def\deltay{0.5} 


\node [dot, fill=\classicalStructColour] (mult) at (0,0) {};

\end{tikzpicture}}\!$-phases and $a \neq 0$ be (the phase induced by) a $\hbox{\begin{tikzpicture} [scale=1,transform shape] 

\def\deltax{0.3} 
\def\deltay{0.5} 


\node [dot, fill=\groupStructColour] (mult) at (0,0) {};

\end{tikzpicture}}\!$-classical point such that the following equation (in additive $\integers$-module notation, for $n_r \in \integers$) has solution $(x_r := a_r)_{r=1,...,M}$ in the group of $\hbox{\begin{tikzpicture} [scale=1,transform shape] 

\def\deltax{0.3} 
\def\deltay{0.5} 


\node [dot, fill=\classicalStructColour] (mult) at (0,0) {};

\end{tikzpicture}}\!$-phases, but has no solution in the subgroup of (phases induced by) $\hbox{\begin{tikzpicture} [scale=1,transform shape] 

\def\deltax{0.3} 
\def\deltay{0.5} 


\node [dot, fill=\groupStructColour] (mult) at (0,0) {};

\end{tikzpicture}}\!$-classical points:
		\begin{equation}\label{eqn_MerminNonLocalityProofEquation}
			\sum_{r=1}^{M} n_r \cdot a_r = a
		\end{equation}
		This means that we are assuming the group of $\hbox{\begin{tikzpicture} [scale=1,transform shape] 

\def\deltax{0.3} 
\def\deltay{0.5} 


\node [dot, fill=\classicalStructColour] (mult) at (0,0) {};

\end{tikzpicture}}\!$-phases are a non-trivial algebraic extension of the subgroup of $\hbox{\begin{tikzpicture} [scale=1,transform shape] 

\def\deltax{0.3} 
\def\deltay{0.5} 


\node [dot, fill=\groupStructColour] (mult) at (0,0) {};

\end{tikzpicture}}\!$-classical points. Without loss of generality, assume that $n_r \neq 0$ and $a_r \neq 0$ for all $r=1,...,M$. 
		
\noindent Let $k$ be the exponent of the group of $\hbox{\begin{tikzpicture} [scale=1,transform shape] 

\def\deltax{0.3} 
\def\deltay{0.5} 


\node [dot, fill=\groupStructColour] (mult) at (0,0) {};

\end{tikzpicture}}\!$-classical points, and define the following Mermin measurement, where each $a_r$ appears $n_r$ times and $0$ appears $n_0$ times, for some $n_0$ such that $V := \sum_{r=0}^{M} n_r \equiv \modclass{1}{k}$ 
		\begin{equation}\label{eqn_MerminNonLocalityProofMeasurement}
			\underline{\alpha} = (a_1,...,a_1,...,a_M,...,a_M,0,...,0)
		\end{equation}

		\noindent Define a $V$-partite Mermin measurement scenario with $S := n_0+V$ and:
		\begin{align*}\label{eqn_MerminNonLocalityProofMeasurement2}
			\underline{\alpha}^s &:= (0,0,...,0,0) \text{ for } s=1,...,n_0\\
			\underline{\alpha}^{n_0+v}_i &:= \underline{\alpha}_{i+\modclass{v}{V}} \text{ for } v=1,...,V
			\numberthis
		\end{align*}
		The scenario has $n_0$ measurements with only $0$ phases (the \textbf{controls}) and $V$ measurements with cyclic permutations of $\underline{\alpha}$ (the \textbf{variations}). The following diagram depicts the scenario:
		\begin{equation}\label{MerminScenarioProof}
		\hbox{\begin{tikzpicture}[transform shape, scale=1.2]

		\node [style=dot, fill={\classicalStructColour}] (0) at (-7.5, -1.75) {};
		\node [style=dot, fill={\classicalStructColour}] (1) at (-6.5, -1.75) {};
		\node [style=dot, fill={\classicalStructColour}, inner sep={6 pt}] (2) at (-8.5, 0.25) {$0$};
		\node [style=dot, fill={\classicalStructColour}, inner sep={6 pt}] (3) at (-7.5, 0.25) {$0$};
		\node [style=dot, fill={\classicalStructColour}, inner sep={6 pt}] (4) at (-6.5, 0.25) {$0$};
		\node [style=dot, fill={\classicalStructColour}, inner sep={6 pt}] (5) at (-5.5, 0.25) {$0$};
		\node [dot, fill={\groupStructColour}, style=none] (6) at (-6, 1) {};
		\node [dot, fill={\groupStructColour}, style=none] (7) at (-8, 1) {};
		\node [style=none] (8) at (-6, 1.75) {};
		\node [style=none] (9) at (-8, 1.75) {};
		\node [style=none] (10) at (-7, 1) {$\cdot\;\cdot\;\cdot$};
		\node [style=none] (11) at (-4.5, -0) {$\cdot\;\cdot\;\cdot$};
		\node [dot, fill={\groupStructColour}, style=none] (12) at (-1, 1) {};
		\node [style=none] (13) at (-2, 1) {$\cdot\;\cdot\;\cdot$};
		\node [style=dot, fill={\classicalStructColour}, inner sep={6 pt}] (14) at (-0.5, 0.25) {$0$};
		\node [style=dot, fill={\classicalStructColour}, inner sep={6 pt}] (15) at (-1.5, 0.25) {$0$};
		\node [style=none] (16) at (-1, 1.75) {};
		\node [dot, fill={\groupStructColour}, style=none] (17) at (-3, 1) {};
		\node [style=none] (18) at (-3, 1.75) {};
		\node [style=dot, fill={\classicalStructColour}] (19) at (-1.5, -1.75) {};
		\node [style=dot, fill={\classicalStructColour}, inner sep={6 pt}] (20) at (-2.5, 0.25) {$0$};
		\node [style=dot, fill={\classicalStructColour}] (21) at (-2.5, -1.75) {};
		\node [style=dot, fill={\classicalStructColour}, inner sep={6 pt}] (22) at (-3.5, 0.25) {$0$};
		\node [style=dot, fill={\classicalStructColour}, inner sep={4 pt}] (23) at (3.5, 0.25) {$\alpha_N^1$};
		\node [style=none] (24) at (7, 1) {$\cdot\;\cdot\;\cdot$};
		\node [style=dot, fill={\classicalStructColour}, inner sep={1 pt}] (25) at (7.5, 0.25) {$-\alpha_N^V$};
		\node [dot, fill={\groupStructColour}, style=none] (26) at (8, 1) {};
		\node [style=dot, fill={\classicalStructColour}] (27) at (6.5, -1.75) {};
		\node [style=none] (28) at (3, 1.75) {};
		\node [style=none] (29) at (8, 1.75) {};
		\node [style=dot, fill={\classicalStructColour}, inner sep={4 pt}] (30) at (1.5, 0.25) {$\alpha_1^1$};
		\node [style=dot, fill={\classicalStructColour}, inner sep={1 pt}] (31) at (5.5, 0.25) {$-\alpha_1^V$};
		\node [style=none] (32) at (2, 1) {$\cdot\;\cdot\;\cdot$};
		\node [style=dot, fill={\classicalStructColour}, inner sep={1 pt}] (33) at (2.5, 0.25) {$-\alpha_N^1$};
		\node [style=dot, fill={\classicalStructColour}] (34) at (1.5, -1.75) {};
		\node [dot, fill={\groupStructColour}, style=none] (35) at (6, 1) {};
		\node [style=dot, fill={\classicalStructColour}] (36) at (7.5, -1.75) {};
		\node [dot, fill={\groupStructColour}, style=none] (37) at (1, 1) {};
		\node [dot, fill={\groupStructColour}, style=none] (38) at (3, 1) {};
		\node [style=dot, fill={\classicalStructColour}, inner sep={1 pt}] (39) at (0.5, 0.25) {$-\alpha_1^1$};
		\node [style=none] (40) at (4.5, -0) {$\cdot\;\cdot\;\cdot$};
		\node [style=dot, fill={\classicalStructColour}] (41) at (2.5, -1.75) {};
		\node [style=dot, fill={\classicalStructColour}, inner sep={4 pt}] (42) at (6.5, 0.25) {$\alpha_1^V$};
		\node [style=dot, fill={\classicalStructColour}, inner sep={4 pt}] (43) at (8.5, 0.25) {$\alpha_N^V$};
		\node [style=none] (44) at (1, 1.75) {};
		\node [style=none] (45) at (6, 1.75) {};

	\begin{pgfonlayer}{background}
		\draw [->-=.5] (2) to (0);
		\draw [->-=.5] (1) to (3);
		\draw [->-=.5] (4) to (0);
		\draw [->-=.5] (1) to (5);
		\draw [->-=.5, bend right=45, looseness=1.00] (5) to (6.center);
		\draw [->-=.5, bend right=45, looseness=1.00] (6.center) to (4);
		\draw [->-=.5, bend right=45, looseness=1.00] (3) to (7.center);
		\draw [->-=.5, bend right=45, looseness=1.00] (7.center) to (2);
		\draw [->-=.5] (7.center) to (9.center);
		\draw [->-=.5] (6.center) to (8.center);
		\draw [->-=.5] (22) to (21);
		\draw [->-=.5] (19) to (20);
		\draw [->-=.5] (15) to (21);
		\draw [->-=.5] (19) to (14);
		\draw [->-=.5, bend right=45, looseness=1.00] (14) to (12.center);
		\draw [->-=.5, bend right=45, looseness=1.00] (12.center) to (15);
		\draw [->-=.5, bend right=45, looseness=1.00] (20) to (17.center);
		\draw [->-=.5, bend right=45, looseness=1.00] (17.center) to (22);
		\draw [->-=.5] (17.center) to (18.center);
		\draw [->-=.5] (12.center) to (16.center);
		\draw [->-=.5] (39) to (34);
		\draw [->-=.5] (41) to (30);
		\draw [->-=.5] (33) to (34);
		\draw [->-=.5] (41) to (23);
		\draw [->-=.5, bend right=45, looseness=1.00] (23) to (38.center);
		\draw [->-=.5, bend right=45, looseness=1.00] (38.center) to (33);
		\draw [->-=.5, bend right=45, looseness=1.00] (30) to (37.center);
		\draw [->-=.5, bend right=45, looseness=1.00] (37.center) to (39);
		\draw [->-=.5] (37.center) to (44.center);
		\draw [->-=.5] (38.center) to (28.center);
		\draw [->-=.5] (31) to (27);
		\draw [->-=.5] (36) to (42);
		\draw [->-=.5] (25) to (27);
		\draw [->-=.5] (36) to (43);
		\draw [->-=.5, bend right=45, looseness=1.00] (43) to (26.center);
		\draw [->-=.5, bend right=45, looseness=1.00] (26.center) to (25);
		\draw [->-=.5, bend right=45, looseness=1.00] (42) to (35.center);
		\draw [->-=.5, bend right=45, looseness=1.00] (35.center) to (31);
		\draw [->-=.5] (35.center) to (45.center);
		\draw [->-=.5] (26.center) to (29.center);
	\end{pgfonlayer}
	                
\draw [thick, decoration={brace, mirror,amplitude=9pt},decorate] (-9,-2.3) to (-0.25,-2.3);
\draw [thick, decoration={brace, mirror,amplitude=9pt},decorate] (0.25,-2.3) to (9,-2.3);
\node at (-4.6, -3.1) {\bf controls};
\node at (4.6, -3.1) {\bf variations};

\end{tikzpicture}}
		\end{equation}	

\noindent To show that the scenario from Equation \ref{MerminScenarioProof} does not admit a local hidden variable:
		\begin{enumerate}
			\item[1a.] we add up (in the group of $\hbox{\begin{tikzpicture} [scale=1,transform shape] 

\def\deltax{0.3} 
\def\deltay{0.5} 


\node [dot, fill=\classicalStructColour] (mult) at (0,0) {};

\end{tikzpicture}}\!$-phases) all the components of each control, using Lemma \ref{thm_MerminMeasurementGHZstate}, and obtain $0$ from each control
			\item[1b.] we add up all the components of each variation, again using Lemma \ref{thm_MerminMeasurementGHZstate}, and obtain $a$ from each variation
			\item[2a.] we add up the result from all the controls, and obtain $\Sigma_C := n_0 \cdot 0 = 0$
			\item[2b.] we add up the result from all variations, and obtain $\Sigma_V := V \cdot a = a$ , using the fact that $a$ is in the subgroup of (phases induced by) $\hbox{\begin{tikzpicture} [scale=1,transform shape] 

\def\deltax{0.3} 
\def\deltay{0.5} 


\node [dot, fill=\groupStructColour] (mult) at (0,0) {};

\end{tikzpicture}}\!$-classical points and $V$ is congruent to $1$ modulo the exponent of the subgroup
			\item[3.] we subtract $\Sigma_C$ from $\Sigma_V$, using the antipode $\hbox{\begin{tikzpicture} [scale=1,transform shape] 

\def\deltax{0.3} 
\def\deltay{0.5} 


\node [antipode] (mult) at (0,0) {};
\node (mult_label_in) at (0,-\deltay) {};
\node (mult_label_out) at (0,+\deltay) {};
\draw[-] (mult_label_in) to (mult);
\draw[-] (mult) to (mult_label_out);

\end{tikzpicture}}\!$ of the strongly complementary pair $(\hbox{\begin{tikzpicture} [scale=1,transform shape] 

\def\deltax{0.3} 
\def\deltay{0.5} 


\node [dot, fill=\classicalStructColour] (mult) at (0,0) {};

\end{tikzpicture}}\!,\hbox{\begin{tikzpicture} [scale=1,transform shape] 

\def\deltax{0.3} 
\def\deltay{0.5} 


\node [dot, fill=\groupStructColour] (mult) at (0,0) {};

\end{tikzpicture}}\!)$, and obtain $a-0 = a$
			\item[4.] we test the result against the $\hbox{\begin{tikzpicture} [scale=1,transform shape] 

\def\deltax{0.3} 
\def\deltay{0.5} 


\node [dot, fill=\groupStructColour] (mult) at (0,0) {};

\end{tikzpicture}}\!$-classical point $\bra{a}$, and obtain the non-zero scalar $\braket{a}{a}$
		\end{enumerate}		
		The procedure is summarised by the following diagram:
		\begin{equation}\label{MerminSetupProof}
			\hbox{\begin{tikzpicture}[node distance = 9.5mm]

\node (a) [dot, inner sep = 1mm, fill = \Zcolour] {$a$};
\node (topdot) [smalldot, below of = a, fill = \Zcolour, yshift = 2mm] {};
\node (antipode) [dot, below left of = topdot,  yshift = 2mm,xshift = -3mm] {};
\node (antipoderfiller) [below right of = topdot, yshift = 2mm,xshift = +6mm] {};

\node (ldot) [smalldot, below left of = antipode, fill = \Zcolour, yshift = 2mm, xshift = -3mm] {};
\node (rdot) [smalldot, below right of = antipoderfiller, fill = \Zcolour, yshift = 2mm, xshift = + 6mm] {};

\node (lcontroldot) [smalldot, below left of = ldot, fill = \Zcolour, yshift = 2mm, xshift = -2mm] {};
\node (lcontrol) [box, below of = lcontroldot, yshift = 2mm] {$0 ... 0$};

\node (rcontroldot) [smalldot, below right of = ldot, fill = \Zcolour, yshift = 2mm,xshift = +2mm] {};
\node (rcontrol) [box, below of = rcontroldot, yshift = 2mm] {$0 ... 0$};

\node (lvariationdot) [smalldot, below left of = rdot, fill = \Zcolour, yshift = 2mm,xshift = -8mm] {};
\node (lvariation) [box, below of = lvariationdot, yshift = 2mm] {$a_1...a_m0...0$};

\node (rvariationdot) [smalldot, below right of = rdot, fill = \Zcolour, yshift = 2mm,xshift = +8mm] {};
\node (rvariation) [box, below of = rvariationdot, yshift = 2mm] {$0a_1...a_m0...0$};

\node (controlellipsis) [below of = ldot, yshift = -2.5mm] {...};
\node (variationellipsis) [below of = rdot, yshift = -2.5mm] {...};
\node (controllabel) [below of = controlellipsis, yshift = 3mm] {$n_0$ controls};
\node (variationlabel) [below of = variationellipsis, yshift = 3mm] {$V$ variations};

\node (lcontrolellipsis) [below of = lcontroldot, yshift = +7mm] {...};
\node (rcontrolellipsis) [below of = rcontroldot, yshift = +7mm] {...};
\node (lvariationellipsis) [below of = lvariationdot, yshift = +7mm] {...};
\node (rvariationellipsis) [below of = rvariationdot, yshift = +7mm] {...};

\begin{pgfonlayer}{background}
\draw[->-=.5,out=90,in=270] (topdot) to (a);
\draw[->-=.5,out=45,in=225] (antipode) to (topdot);
\draw[->-=.5,out=45,in=225] (ldot) to (antipode);
\draw[->-=.5,out=135,in=315] (rdot) to (topdot);

\draw[->-=.5,out=45,in=225] (lcontroldot) to (ldot);
\draw[->-=.5,out=135,in=315] (rcontroldot) to (ldot);
\draw[->-=.5,out=45,in=225] (lvariationdot) to (rdot);
\draw[->-=.5,out=135,in=315] (rvariationdot) to (rdot);

\draw[->-=.5,out=90,in=225] (lcontrol.135) to (lcontroldot);
\draw[->-=.5,out=90,in=315] (lcontrol.45) to (lcontroldot);
\draw[->-=.5,out=90,in=225] (rcontrol.135) to (rcontroldot);
\draw[->-=.5,out=90,in=315] (rcontrol.45) to (rcontroldot);

\draw[->-=.5,out=90,in=225] (lvariation.155) to (lvariationdot);
\draw[->-=.5,out=90,in=315] (lvariation.25) to (lvariationdot);
\draw[->-=.5,out=90,in=225] (rvariation.155) to (rvariationdot);
\draw[->-=.5,out=90,in=315] (rvariation.25) to (rvariationdot);
\end{pgfonlayer}


\node (l0) [left of = ldot, xshift = +3mm, yshift = -1mm] {$0$};
\node (r0) [right of = ldot, xshift = -3mm, yshift = -1mm] {$0$};

\node (la) [left of = rdot, xshift = +2mm, yshift = 0mm] {$a$};
\node (ra) [right of = rdot, xshift = -2mm, yshift = 0mm] {$a$};

\node (n00) [left of = antipode, xshift = +3mm, yshift = 0mm] {$n_0 \cdot 0$};
\node (vaa) [right of = topdot, xshift = 0mm, yshift = -2mm] {$V \cdot a$};

\node (topa) [left of = topdot, xshift = +7mm, yshift = 3.5mm] {$a$};

\end{tikzpicture}}
		\end{equation}

\noindent The same procedure applied to any local hidden variable model always yields the $0$ scalar. A local hidden variable model is nothing but the local map for the scenario applied to some state, so it is enough to show that the above procedure yields the constant $0$ function when composed with the local map:
		\begin{equation}\label{MerminSetupProofLHV}
		\hbox{\begin{tikzpicture}[node distance = 9.5mm]

\node (a) [dot, inner sep = 1mm, fill = \Zcolour] {$a$};
\node (topdot) [smalldot, below of = a, fill = \Zcolour, yshift = 2mm] {};
\node (antipode) [dot, below left of = topdot,  yshift = 2mm,xshift = -3mm] {};
\node (antipoderfiller) [below right of = topdot, yshift = 2mm,xshift = +6mm] {};

\node (ldot) [smalldot, below left of = antipode, fill = \Zcolour, yshift = 2mm, xshift = -3mm] {};
\node (rdot) [smalldot, below right of = antipoderfiller, fill = \Zcolour, yshift = 2mm, xshift = + 6mm] {};

\node (lcontroldot) [smalldot, below left of = ldot, fill = \Zcolour, yshift = 2mm, xshift = -2mm] {};
\node (lcontrol) [box, below of = lcontroldot, yshift = 2mm] {};

\node (rcontroldot) [smalldot, below right of = ldot, fill = \Zcolour, yshift = 2mm,xshift = +2mm] {};
\node (rcontrol) [box, below of = rcontroldot, yshift = 2mm] {};

\node (lvariationdot) [smalldot, below left of = rdot, fill = \Zcolour, yshift = 2mm,xshift = -8mm] {};
\node (lvariation) [box, below of = lvariationdot, yshift = 2mm] {};

\node (rvariationdot) [smalldot, below right of = rdot, fill = \Zcolour, yshift = 2mm,xshift = +8mm] {};
\node (rvariation) [box, below of = rvariationdot, yshift = 2mm] {};

\node (localmap) [box, below of = a, yshift = -21mm, xshift=6mm, yscale=1.6, xscale=14] {};

\node (localmaplabel) [below of = a, yshift = -21mm, xshift=6mm] {\bf Local Map};


\node (lcontrolellipsis) [below of = lcontroldot, yshift = +7mm] {...};
\node (rcontrolellipsis) [below of = rcontroldot, yshift = +7mm] {...};
\node (lvariationellipsis) [below of = lvariationdot, yshift = +7mm] {...};
\node (rvariationellipsis) [below of = rvariationdot, yshift = +7mm] {...};

\node (xxx) [below of = rvariationdot, yshift = -9mm, xshift=-10mm] {.\;.\;.};
\node (L2) [below of = rvariationdot, yshift = -9mm, xshift=-16mm] {$\alpha_1$};
\node (L3) [below of = rvariationdot, yshift = -9mm, xshift=-3mm] {$\alpha_M$};
\node (xxxx) [below of = rvariationdot, yshift = -13mm, xshift=-9mm] {System $N$};

\node (x) [below of = rcontrolellipsis, yshift = -6mm, xshift=-10mm] {.\;.\;.};
\node (L0) [below of = rcontrolellipsis, yshift = -6mm, xshift=-16mm] {$\alpha_1$};
\node (L1) [below of = rcontrolellipsis, yshift = -6mm, xshift=-3mm] {$\alpha_M$};
\node (xx) [below of = rcontrolellipsis, yshift = -10mm, xshift=-9mm] {System $1$};

\node (xxxxx) [below of = a, yshift = -33mm, xshift=5mm] {.\quad.\quad.\quad.};

\begin{pgfonlayer}{background}
\draw[->-=.5,out=90,in=270] (topdot) to (a);
\draw[->-=.5,out=45,in=225] (antipode) to (topdot);
\draw[->-=.5,out=45,in=225] (ldot) to (antipode);
\draw[->-=.5,out=135,in=315] (rdot) to (topdot);

\draw[->-=.5,out=45,in=225] (lcontroldot) to (ldot);
\draw[->-=.5,out=135,in=315] (rcontroldot) to (ldot);
\draw[->-=.5,out=45,in=225] (lvariationdot) to (rdot);
\draw[->-=.5,out=135,in=315] (rvariationdot) to (rdot);

\draw[->-=.5,out=90,in=225] (lcontrol.135) to (lcontroldot);
\draw[->-=.5,out=90,in=315] (lcontrol.45) to (lcontroldot);
\draw[->-=.5,out=90,in=225] (rcontrol.135) to (rcontroldot);
\draw[->-=.5,out=90,in=315] (rcontrol.45) to (rcontroldot);

\draw[->-=.5,out=90,in=225] (lvariation.135) to (lvariationdot);
\draw[->-=.5,out=90,in=315] (lvariation.45) to (lvariationdot);
\draw[->-=.5,out=90,in=225] (rvariation.135) to (rvariationdot);
\draw[->-=.5,out=90,in=315] (rvariation.45) to (rvariationdot);

\draw[->-=.5] (L0) to (-5.35,-7.1);
\draw[->-=.5] (L1) to (-2.8,-7.1);
\draw[->-=.5] (L2) to (5.05,-7.1);
\draw[->-=.5] (L3) to (7.6,-7.1);
\end{pgfonlayer}


\node (l0) [left of = ldot, xshift = +3mm, yshift = -1mm] {$0$};
\node (r0) [right of = ldot, xshift = -3mm, yshift = -1mm] {$0$};

\node (la) [left of = rdot, xshift = +2mm, yshift = 0mm] {$a$};
\node (ra) [right of = rdot, xshift = -2mm, yshift = 0mm] {$a$};

\node (n00) [left of = antipode, xshift = +3mm, yshift = 0mm] {$n_0 \cdot 0$};
\node (vaa) [right of = topdot, xshift = 0mm, yshift = -2mm] {$V \cdot a$};

\node (topa) [left of = topdot, xshift = +7mm, yshift = 3.5mm] {$a$};

\end{tikzpicture}}
		\end{equation}		
		Since the $\hbox{\begin{tikzpicture} [scale=1,transform shape] 

\def\deltax{0.3} 
\def\deltay{0.5} 


\node [dot, fill=\groupStructColour] (mult) at (0,0) {};

\end{tikzpicture}}\!$-classical points form a basis, it is sufficient to show that the map from Diagram \ref{MerminSetupProofLHV} always yields $0$ when applied to $\hbox{\begin{tikzpicture} [scale=1,transform shape] 

\def\deltax{0.3} 
\def\deltay{0.5} 


\node [dot, fill=\groupStructColour] (mult) at (0,0) {};

\end{tikzpicture}}\!$-classical points. In the following diagram, the $\hbox{\begin{tikzpicture} [scale=1,transform shape] 

\def\deltax{0.3} 
\def\deltay{0.5} 


\node [dot, fill=\classicalStructColour] (mult) at (0,0) {};

\end{tikzpicture}}\!$ nodes have been re-arranged using the spider theorem, so that the wiring of the local map can be written down explicitly in a clean way. The diagram also annotates the $\hbox{\begin{tikzpicture} [scale=1,transform shape] 

\def\deltax{0.3} 
\def\deltay{0.5} 


\node [dot, fill=\groupStructColour] (mult) at (0,0) {};

\end{tikzpicture}}\!$-classical values on the wires at each stage to aid in following the argument:
		\begin{enumerate}
			\item[1.] the values $b_0^1,...,b_0^V$ for the $0$ phases of systems $1$ to $V$ are each duplicated $n_0+n_0$ times and then added up to $b_0 := n_0 \cdot \sum{i=1}^V b_0^i$ by the two $\hbox{\begin{tikzpicture} [scale=1,transform shape] 

\def\deltax{0.3} 
\def\deltay{0.5} 


\node [dot, fill=\classicalStructColour] (mult) at (0,0) {};

\end{tikzpicture}}\!$ nodes
			\item[2.] the values $b_1^i,...,b_m^i$ for the $a_1,...,a_m$ phases of each system $i=1$ (for $i=1,...,V$) are each duplicated $n_k$ times (for $k=1,...,m$) and added up to $b^i := \sum_{r=1}^m n_r \cdot b_r^i$ by the respective $\hbox{\begin{tikzpicture} [scale=1,transform shape] 

\def\deltax{0.3} 
\def\deltay{0.5} 


\node [dot, fill=\classicalStructColour] (mult) at (0,0) {};

\end{tikzpicture}}\!$ nodes
			\item[3.] the values $b^1,...,b^V$ are added up to $b := \sum_{i=1}^V b^i$
			\item[4.] the value $b_0$ is added up to $b$
			\item[5.] finally, the value $b_0$ is subtracted from $b$, and $b$ is tested against the $\hbox{\begin{tikzpicture} [scale=1,transform shape] 

\def\deltax{0.3} 
\def\deltay{0.5} 


\node [dot, fill=\groupStructColour] (mult) at (0,0) {};

\end{tikzpicture}}\!$-classical point $\bra{a}$, obtaining the scalar $\braket{a}{b}$ (which we want to be zero)
		\end{enumerate} 
        The steps are summarised by the following diagram:
		\begin{equation}\label{MerminSetupProofLHVeval}
			\hbox{\begin{tikzpicture}[node distance = 10mm]

\node (a) [dot, inner sep = 1mm, fill = \Zcolour] {$a$};
\node (topdot) [smalldot, below of = a, fill = \Zcolour, yshift = 2mm] {};
\node (antipode) [dot, below left of = topdot,  yshift = 0mm,xshift = -3mm] {};
\node (antipodefiller) [smalldot, fill = \Zcolour, below right of = topdot, yshift = 0mm,xshift = +6mm] {};

\node (ldot) [smalldot, below left of = antipode, fill = \Zcolour, yshift = 0mm, xshift = -3mm] {};
\node (rdot) [smalldot, below right of = antipodefiller, fill = \Zcolour, yshift = 0mm, xshift = + 6mm] {};

\node (lcontroldot) [smalldot, below left of = ldot, fill = \Zcolour, yshift = -2mm, xshift = -2mm] {};
\node (l0phasedot) [smalldot, below of = lcontroldot, yshift = 2mm, fill = \Xcolour] {};
\node (l0phase) [smalldot, inner sep = 1mm, below of = l0phasedot, yshift = 1mm, fill = \Zcolour] {$b_0^1$};

\node (rcontroldot) [smalldot, below right of = ldot, fill = \Zcolour, yshift = -2mm,xshift = +2mm] {};
\node (r0phasedot) [smalldot, below of = rcontroldot, yshift = 2mm, fill = \Xcolour] {};
\node (r0phase) [dot, inner sep = 1mm, below of = r0phasedot, yshift = 1mm, fill = \Zcolour] {$b_0^V$};

\node (lvariationdot) [smalldot, below left of = rdot, fill = \Zcolour, yshift = 0mm,xshift = -8mm] {};
\node (la1variationdot) [smalldot, below left of = lvariationdot, yshift = -3mm, fill = \Xcolour] {};
\node (la1phase) [dot, inner sep = 1mm, below of = la1variationdot, yshift = 1mm, fill = \Zcolour] {$b_1^1$};
\node (laMvariationdot) [smalldot, below right of = lvariationdot, yshift = -3mm, fill = \Xcolour] {};
\node (laMphase) [dot, inner sep = 1mm, below of = laMvariationdot, yshift = 1mm, fill = \Zcolour] {$b_M^1$};

\node (rvariationdot) [smalldot, below right of = rdot, fill = \Zcolour, yshift = 0mm,xshift = +8mm] {};
\node (ra1variationdot) [smalldot, below left of = rvariationdot, yshift = -3mm, fill = \Xcolour] {};
\node (ra1phase) [dot, inner sep = 1mm, below of = ra1variationdot, yshift = 1mm, fill = \Zcolour] {$b_1^V$};
\node (raMvariationdot) [smalldot, below right of = rvariationdot, yshift = -3mm, fill = \Xcolour] {};
\node (raMphase) [dot, inner sep = 1mm, below of = raMvariationdot, yshift = 1mm, fill = \Zcolour] {$b_M^V$};

\node (controlellipsis) [below of = ldot, yshift = -10mm] {...};
\node (variationellipsis) [below of = rdot, yshift = -2.5mm] {...};

\node (lcontrolellipsislabel) [below of = lcontroldot, yshift = +6.5mm] {$n_0$};
\node (lcontrolellipsis) [below of = lcontroldot, yshift = +4.5mm] {$...$};
\node (rcontrolellipsislabel) [below of = rcontroldot, yshift = +6.5mm] {$n_0$};
\node (rcontrolellipsis) [below of = rcontroldot, yshift = +4.5mm] {$...$};

\node (rlcontrolellipsislabelcross) [below of = ldot, yshift = -2.75mm] {$n_0$};
\node (rlcontrolellipsiscross) [below of = ldot, yshift = -4.25mm] {$...$};

\node (la1variationellipsis) [below of = lvariationdot, yshift = +3mm, xshift = -6.5mm] {$...$};
\node (la1variationellipsislabel) [below of = lvariationdot, yshift = +4.75mm, xshift = -6.5mm] {$n_1$};
\node (laMvariationellipsis) [below of = lvariationdot, yshift = +3mm, xshift = +6.5mm] {$...$};
\node (laMvariationellipsislabel) [below of = lvariationdot, yshift = +4.75mm, xshift = +6.5mm] {$n_M$};

\node (ra1variationellipsis) [below of = rvariationdot, yshift = +3mm, xshift = -6.5mm] {$...$};
\node (ra1variationellipsislabel) [below of = rvariationdot, yshift = +4.75mm, xshift = -6.5mm] {$n_1$};
\node (raMvariationellipsis) [below of = rvariationdot, yshift = +3mm, xshift = +6.5mm] {$...$};
\node (raMvariationellipsislabel) [below of = rvariationdot, yshift = +4.75mm, xshift = +6.5mm] {$n_M$};

\node (rvariationellipsis) [below of = rvariationdot, yshift = +6mm] {$...$};

\begin{pgfonlayer}{background}
\draw[->-=.5,out=90,in=270] (topdot) to (a);
\draw[->-=.5,out=45,in=225] (antipode) to (topdot);
\draw[->-=.5,out=45,in=225] (ldot) to (antipode);
\draw[->-=.5,out=135,in=315] (rdot) to (antipodefiller);
\draw[->-=.5,out=135,in=315] (antipodefiller) to (topdot);

\draw[->-=.5,out=45,in=225] (lcontroldot) to (ldot);
\draw[->-=.5,out=45,in=225] (rcontroldot) to (antipodefiller);
\draw[->-=.5,out=45,in=225] (lvariationdot) to (rdot);
\draw[->-=.5,out=135,in=315] (rvariationdot) to (rdot);

\draw[->-=.5,out=135,in=225] (l0phasedot) to (lcontroldot);
\draw[->-=.5,out=45,in=315] (l0phasedot) to (lcontroldot);
\draw[->-=.5,out=135,in=225] (r0phasedot) to (rcontroldot);
\draw[->-=.5,out=45,in=315] (r0phasedot) to (rcontroldot);

\draw[->-=.5,out=45,in=180] (l0phasedot) to (rcontroldot);
\draw[->-=.5,out=0,in=225] (l0phasedot) to (rcontroldot);
\draw[->-=.5,out=135,in=0] (r0phasedot) to (lcontroldot);
\draw[->-=.5,out=180,in=315] (r0phasedot) to (lcontroldot);

\draw[->-=.5,out=90,in=270] (l0phase) to (l0phasedot);
\draw[->-=.5,out=90,in=270] (r0phase) to (r0phasedot);
\draw[->-=.5,out=90,in=270] (la1phase) to (la1variationdot);
\draw[->-=.5,out=90,in=270] (laMphase) to (laMvariationdot);
\draw[->-=.5,out=90,in=270] (ra1phase) to (ra1variationdot);
\draw[->-=.5,out=90,in=270] (raMphase) to (raMvariationdot);

\draw[->-=.5,out=135,in=180] (la1variationdot) to (lvariationdot);
\draw[->-=.5,out=45,in=225] (la1variationdot) to (lvariationdot);
\draw[->-=.5,out=135,in=315] (laMvariationdot) to (lvariationdot);
\draw[->-=.5,out=45,in=0] (laMvariationdot) to (lvariationdot);

\draw[->-=.5,out=135,in=180] (ra1variationdot) to (rvariationdot);
\draw[->-=.5,out=45,in=225] (ra1variationdot) to (rvariationdot);
\draw[->-=.5,out=135,in=315] (raMvariationdot) to (rvariationdot);
\draw[->-=.5,out=45,in=0] (raMvariationdot) to (rvariationdot);

\end{pgfonlayer}


\node (lb0) [left of = ldot, xshift = +3.5mm, yshift = -2mm] {$b_0$};
\node (rb0) [right of = ldot, xshift = +12mm, yshift = +2mm] {$b_0$};

\node (la) [left of = rdot, xshift = +2mm, yshift = -1mm] {$b^1$};
\node (ra) [right of = rdot, xshift = -2mm, yshift = -1mm] {$b^V$};
\node (ra) [right of = antipodefiller, xshift = -2mm, yshift = -1mm] {$b$};

\node (n00) [left of = antipode, xshift = +10mm, yshift = 5mm] {$-b_0$};
\node (vaa) [right of = topdot, xshift = 0mm, yshift = -1mm] {$b_0+b$};

\node (topa) [left of = topdot, xshift = +7mm, yshift = 3.5mm] {$b$};

\end{tikzpicture}}
		\end{equation}	

\noindent The $\hbox{\begin{tikzpicture} [scale=1,transform shape] 

\def\deltax{0.3} 
\def\deltay{0.5} 


\node [dot, fill=\groupStructColour] (mult) at (0,0) {};

\end{tikzpicture}}\!$-classical points $c$ that can be written as $c = \sum_{r=1}^M n_r \cdot c_r$ for some $\hbox{\begin{tikzpicture} [scale=1,transform shape] 

\def\deltax{0.3} 
\def\deltay{0.5} 


\node [dot, fill=\groupStructColour] (mult) at (0,0) {};

\end{tikzpicture}}\!$-classical points $c_1,...,c_M$ form a subgroup $H$ of the group of $\hbox{\begin{tikzpicture} [scale=1,transform shape] 

\def\deltax{0.3} 
\def\deltay{0.5} 


\node [dot, fill=\groupStructColour] (mult) at (0,0) {};

\end{tikzpicture}}\!$-classical points. Indeed we have that $0 = \sum_{r=1}^m n_r \cdot 0$ and that $(\sum_{r=1}^M n_r \cdot c_r)+(\sum_{r=1}^M n_r \cdot d_r) = \sum_{r=1}^M n_r \cdot (c_r+d_r)$. Furthermore, by assumption we have that $H$ does not contain $a$, and as a consequence $\braket{a}{c} = 0$ for all $c \in H$. Going back to Diagram \ref{MerminSetupProofLHVeval}, we see that $b^1,...,b^V \in H$ (but $b_0$ need not be in $H$, hence the need to subtract it before testing against $a$). We thus conclude that $b\in H$ (since $H$ is closed under addition): hence the scalar $\braket{a}{b}$ vanishes, concluding our proof that no local hidden variable can exist for our chosen measurement scenario. 
	\end{proof}

\begin{theorem}[Mermin Locality]
	\label{thm_MerminLocality}
		Let $\CategoryC$ be a $\dagger$-SMC. If for any strongly complementary pair $(\hbox{\begin{tikzpicture} [scale=1,transform shape] 

\def\deltax{0.3} 
\def\deltay{0.5} 


\node [dot, fill=\classicalStructColour] (mult) at (0,0) {};

\end{tikzpicture}}\!,\hbox{\begin{tikzpicture} [scale=1,transform shape] 

\def\deltax{0.3} 
\def\deltay{0.5} 


\node [dot, fill=\groupStructColour] (mult) at (0,0) {};

\end{tikzpicture}}\!)$ of $\dagger$-qSCFAs the group of $\hbox{\begin{tikzpicture} [scale=1,transform shape] 

\def\deltax{0.3} 
\def\deltay{0.5} 


\node [dot, fill=\classicalStructColour] (mult) at (0,0) {};

\end{tikzpicture}}\!$-phases is a trivial algebraic extension of the subgroup of $\hbox{\begin{tikzpicture} [scale=1,transform shape] 

\def\deltax{0.3} 
\def\deltay{0.5} 


\node [dot, fill=\groupStructColour] (mult) at (0,0) {};

\end{tikzpicture}}\!$-classical points (i.e. if there exist no algebraically non-trivial $\hbox{\begin{tikzpicture} [scale=1,transform shape] 

\def\deltax{0.3} 
\def\deltay{0.5} 


\node [dot, fill=\classicalStructColour] (mult) at (0,0) {};

\end{tikzpicture}}\!$-phases), then $\CategoryC$ is Mermin local.
	\end{theorem}
	\begin{proof} 
		Consider an $N$-partite Mermin measurement scenario $\underline{\alpha}^s = (\alpha_1^s,...,\alpha_N^s)_{s=1,...,S}$, and let $a_1,...,a_M$ be the distinct $\hbox{\begin{tikzpicture} [scale=1,transform shape] 

\def\deltax{0.3} 
\def\deltay{0.5} 


\node [dot, fill=\classicalStructColour] (mult) at (0,0) {};

\end{tikzpicture}}\!$-phases appearing in it. Consider the system of equations $(\sum_{r=1}^M n_r^s \cdot x_r = c^s)_{s=1,...,S}$, where $n_r^s$ is the numer of times phase $a_r$ appears in measurement $\underline{\alpha}^s$, and $c^s$ are the unique values making $x_r := a_r$ into a solution for the system. As the group of $\hbox{\begin{tikzpicture} [scale=1,transform shape] 

\def\deltax{0.3} 
\def\deltay{0.5} 


\node [dot, fill=\classicalStructColour] (mult) at (0,0) {};

\end{tikzpicture}}\!$-phases is a trivial algebraic extension of the subgroup of $\hbox{\begin{tikzpicture} [scale=1,transform shape] 

\def\deltax{0.3} 
\def\deltay{0.5} 


\node [dot, fill=\groupStructColour] (mult) at (0,0) {};

\end{tikzpicture}}\!$-classical points, there is a solution $x_r := b_r$ with $(b_r)_{r=1,...,M}$ $\hbox{\begin{tikzpicture} [scale=1,transform shape] 

\def\deltax{0.3} 
\def\deltay{0.5} 


\node [dot, fill=\groupStructColour] (mult) at (0,0) {};

\end{tikzpicture}}\!$-classical points. By using this, together with Lemma \ref{thm_MerminMeasurementGHZstate}, we see that each measurement in the scenario is equal to the Mermin measurement obtained by replacing $a_r$ with $b_r$ for all $r=1,...,M$ (say $\beta_i^s := b_r$ if $\alpha_i^s = a_r$): 
		\begin{equation}\label{MerminScenarioSwap}
		    \hbox{\begin{tikzpicture}[node distance=10mm, xscale=1.5, yscale=1.5]

                \node [smalldot, fill={\classicalStructColour}] (0) at (-3.75, -1.25) {};
                \node [smalldot, fill={\classicalStructColour}] (1) at (-2.75, -1.25) {};
                \node [style=dot, fill={\classicalStructColour}, inner sep={0 pt}] (2) at (-5.45, 0.25) {$-\alpha_1^s$};
                \node [style=dot, fill={\classicalStructColour}, inner sep={3 pt}] (3) at (-4.05, 0.25) {$\alpha_1^s$};
                \node [style=dot, fill={\classicalStructColour}, inner sep={0 pt}] (4) at (-2.45, 0.25) {$-\alpha_N^s$};
                \node [style=dot, fill={\classicalStructColour}, inner sep={2.4 pt}] (5) at (-1.05, 0.25) {$\alpha_N^s$};
                \node [smalldot, fill={\groupStructColour}, style=none] (6) at (-1.75, 1.25) {};
                \node [smalldot, fill={\groupStructColour}, style=none] (7) at (-4.75, 1.25) {};
                \node (8) at (-1.75, 2) {};
                \node (9) at (-4.75, 2) {};
                \node (10) at (-3.25, 1) {$\cdot\;\cdot\;\cdot$};
                \node (11) at (3.25, 1) {$\cdot\;\cdot\;\cdot$};
                \node [style=dot, fill={\classicalStructColour}, inner sep={0 pt}] (12) at (1.05, 0.25) {$-\beta_1^s$};
                \node [smalldot, fill={\classicalStructColour}] (13) at (2.75, -1.25) {};
                \node [smalldot, fill={\groupStructColour}, style=none] (14) at (4.75, 1.25) {};
                \node [smalldot, fill={\groupStructColour}, style=none] (15) at (1.75, 1.25) {};
                \node [style=dot, fill={\classicalStructColour}, inner sep={3 pt}] (16) at (2.45, 0.25) {$\beta_1^s$};
                \node [style=dot, fill={\classicalStructColour}, inner sep={3 pt}] (17) at (5.45, 0.25) {$\beta_N^s$};
                \node (18) at (1.75, 2) {};
                \node [smalldot, fill={\classicalStructColour}] (19) at (3.75, -1.25) {};
                \node (20) at (4.75, 2) {};
                \node [style=dot, fill={\classicalStructColour}, inner sep={0 pt}] (21) at (4.05, 0.25) {$-\beta_N^s$};
                \node (22) at (0, -0) {$=$};

\begin{pgfonlayer}{background}
                \draw [->-=.5] (2) to (0);
                \draw [->-=.5] (1) to (3);
                \draw [->-=.5] (4) to (0);
                \draw [->-=.5] (1) to (5);
                \draw [->-=.5, bend right=45, looseness=1.00] (5) to (6.center);
                \draw [->-=.5, bend right=45, looseness=1.00] (6.center) to (4);
                \draw [->-=.5, bend right=45, looseness=1.00] (3) to (7.center);
                \draw [->-=.5, bend right=45, looseness=1.00] (7.center) to (2);
                \draw [->-=.5] (7.center) to (9.center);
                \draw [->-=.5] (6.center) to (8.center);
                \draw [->-=.5] (12) to (13);
                \draw [->-=.5] (19) to (16);
                \draw [->-=.5] (21) to (13);
                \draw [->-=.5] (19) to (17);
                \draw [->-=.5, bend right=45, looseness=1.00] (17) to (14.center);
                \draw [->-=.5, bend right=45, looseness=1.00] (14.center) to (21);
                \draw [->-=.5, bend right=45, looseness=1.00] (16) to (15.center);
                \draw [->-=.5, bend right=45, looseness=1.00] (15.center) to (12);
                \draw [->-=.5] (15.center) to (18.center);
                \draw [->-=.5] (14.center) to (20.center);
\end{pgfonlayer}

\end{tikzpicture}}
		\end{equation}
		All phases are now induced by $\hbox{\begin{tikzpicture} [scale=1,transform shape] 

\def\deltax{0.3} 
\def\deltay{0.5} 


\node [dot, fill=\groupStructColour] (mult) at (0,0) {};

\end{tikzpicture}}\!$-classical points, and can thus be pushed up through the $\XmultSym$s:
		\begin{equation}\label{MerminScenarioPush}
		    \hbox{\begin{tikzpicture}[node distance=10mm, xscale=1.6, yscale=1.5]

                \node (0) at (0.3, -0) {$=$};
                \node [smalldot, fill={\classicalStructColour}] (1) at (-2.5, -1.25) {};
                \node (2) at (-1.5, 2.25) {};
                \node [style=dot, fill={\classicalStructColour}, inner sep={3 pt}] (3) at (-3.8, 0.25) {$\beta_1^s$};
                \node [style=none] (4) at (-4.5, 2.25) {};
                \node [style=dot, fill={\classicalStructColour}, inner sep={0 pt}] (5) at (-5.2, 0.25) {$-\beta_1^s$};
                \node [smalldot, fill={\groupStructColour}, style=none] (6) at (-4.5, 1.25) {};
                \node (7) at (-3, 1.2) {$...$};
                \node [smalldot, fill={\classicalStructColour}] (8) at (-3.5, -1.25) {};
                \node [style=dot, fill={\classicalStructColour}, inner sep={3 pt}] (9) at (-0.7, 0.25) {$\beta_N^s$};
                \node [smalldot, fill={\groupStructColour}, style=none] (10) at (-1.5, 1.25) {};
                \node [style=dot, fill={\classicalStructColour}, inner sep={0 pt}] (11) at (-2.2, 0.25) {$-\beta_N^s$};
                \node [smalldot, fill={\classicalStructColour}] (12) at (1.5, -1.25) {};
                \node [smalldot, fill={\classicalStructColour}] (13) at (2.5, -1.25) {};
                \node [smalldot, fill={\groupStructColour}] (14) at (1, -0.25) {};
                \node [smalldot, fill={\groupStructColour}] (15) at (3, -0.25) {};
                \node [style=dot, fill={\classicalStructColour}, inner sep={3 pt}] (16) at (3, 1) {$\beta_1^s$};
                \node [style=dot, fill={\classicalStructColour}, inner sep={3 pt}] (17) at (1, 1) {$\beta_N^s$};
                \node (18) at (3, 2.25) {};
                \node (19) at (1, 2.25) {};
                \node (20) at (2, 1) {$...$};

        \begin{pgfonlayer}{background}
                \draw [->-=.5] (5) to (8);
                \draw [->-=.5] (1) to (3);
                \draw [->-=.5] (11) to (8);
                \draw [->-=.5] (1) to (9);
                \draw [->-=.5, bend right=45, looseness=1.00] (9) to (10.center);
                \draw [->-=.5, bend right=45, looseness=1.00] (10.center) to (11);
                \draw [->-=.5, bend right=45, looseness=1.00] (3) to (6.center);
                \draw [->-=.5, bend right=45, looseness=1.00] (6.center) to (5);
                \draw [->-=.5] (6.center) to (4.center);
                \draw [->-=.5] (10.center) to (2.center);
                \draw [->-=.5] (14) to (12);
                \draw [->-=.5] (13) to (14);
                \draw [->-=.5] (15) to (12);
                \draw [->-=.5] (13) to (15);
                \draw [->-=.5] (15) to (16);
                \draw [->-=.5] (14) to (17);
                \draw [->-=.5] (17) to (19.center);
                \draw [->-=.5] (16) to (18.center);
        \end{pgfonlayer}
\end{tikzpicture}}
		\end{equation}
		Now that each measurement of the scenario amounts to performing some set of $\hbox{\begin{tikzpicture} [scale=1,transform shape] 

\def\deltax{0.3} 
\def\deltay{0.5} 


\node [dot, fill=\groupStructColour] (mult) at (0,0) {};

\end{tikzpicture}}\!$-classical operations on the same state, it is no surprise that the following gives a local hidden variable model:
		\begin{equation}\label{LHVproof}
		\hbox{\begin{tikzpicture}[xscale=1.6, yscale=1.6]
          \node [smalldot,  fill=\classicalStructColour] (0) at (-0.75, -3) {};
          \node [smalldot,  fill=\classicalStructColour] (1) at (0.75, -3) {};
          \node [smalldot,  fill=\groupStructColour] (2) at (1.75, -2.25) {};
          \node [smalldot,  fill=\groupStructColour] (3) at (-1.75, -2.25) {};
          \node [smalldot,  fill=\groupStructColour] (4) at (-1.75, -1.5) {};
          \node [smalldot,  fill=\groupStructColour] (5) at (1.75, -1.5) {};
          \node [dot,  fill=\classicalStructColour, inner sep=2] (6) at (-2.5, -0.5) {$\beta_1$};
          \node [dot,  fill=\classicalStructColour, inner sep=1] (7) at (-1, -0.5) {$\beta_M$};
          \node [dot,  fill=\classicalStructColour, inner sep=2] (8) at (1, -0.5) {$\beta_1$};
          \node [dot,  fill=\classicalStructColour, inner sep=1] (9) at (2.5, -0.5) {$\beta_M$};
                \node [style=box, xscale=7, yscale=1.7] (10) at (0, 1) {};
                \node at (0, 1) {\bf Local Map};
                \node (11) at (-2.5, 0.5) {};
                \node (12) at (-1, 0.5) {};
                \node (13) at (1, 0.5) {};
                \node (14) at (2.5, 0.5) {};
                \node (15) at (-1, 1.65) {};
                \node (16) at (1, 1.65) {};
                \node (17) at (2.5, 1.65) {};
                \node (18) at (-2.5, 1.65) {};
                \node (19) at (2.5, 2.5) {};
                \node (20) at (1, 2.5) {};
                \node (21) at (-1, 2.5) {};
                \node (22) at (-2.5, 2.5) {};
                \node (23) at (0, -1.5) {$\cdot\;\cdot\;\cdot$};
                \node (24) at (0, -0.5) {$\cdot\cdot\cdot$};
                \node (25) at (-1.75, -0.5) {$...$};
                \node (26) at (1.75, -0.5) {$...$};
                \node (27) at (-1.75, 2.25) {$\cdot\;\cdot\;\cdot$};
                \node (28) at (1.75, 2.25) {$\cdot\;\cdot\;\cdot$};

        \begin{pgfonlayer}{background}
                \draw [->-=.5] (3) to (0);
                \draw [->-=.5] (1) to (3);
                \draw [->-=.5] (2) to (0);
                \draw [->-=.5] (1) to (2);
                \draw [->-=.5] (3) to (4);
                \draw [->-=.5] (4) to (6);
                \draw [->-=.5] (4) to (7);
                \draw [->-=.5] (2) to (5);
                \draw [->-=.5] (5) to (8);
                \draw [->-=.5] (5) to (9);
                \draw [->-=.5] (9) to (14.center);
                \draw [->-=.5] (8) to (13.center);
                \draw [->-=.5] (7) to (12.center);
                \draw [->-=.5] (6) to (11.center);
                \draw [->-=.5] (18.center) to (22.center);
                \draw [->-=.5] (15.center) to (21.center);
                \draw [->-=.5] (16.center) to (20.center);
                \draw [->-=.5] (17.center) to (19.center);
        \end{pgfonlayer}
        
\node at (-3.3, -1.6) {system 1};
\node at (3.3, -1.6) {system N};
        
\end{tikzpicture}}
		\end{equation}	

	\end{proof}

\noindent The abstract framework can now be applied to some particular examples of interest.

	\begin{corollary} The restricted ZX calculus (that corresponds to qubit stabilizer quantum mechanics) from \cite{CQM-ZXCalculusSeminal,CQM-ZXCalculusComplete} (referred to as Stab in \cite{CQM-StrongComplementarity}) is Mermin non-local.
	\end{corollary}
	\begin{proof} 
		Take $\hbox{\begin{tikzpicture} [scale=1,transform shape] 

\def\deltax{0.3} 
\def\deltay{0.5} 


\node [dot, fill=\classicalStructColour] (mult) at (0,0) {};

\end{tikzpicture}}\!$ and $\hbox{\begin{tikzpicture} [scale=1,transform shape] 

\def\deltax{0.3} 
\def\deltay{0.5} 


\node [dot, fill=\groupStructColour] (mult) at (0,0) {};

\end{tikzpicture}}\!$ to be the $Z$ and $X$ single-qubit observables in the ZX calculus. The group of $\hbox{\begin{tikzpicture} [scale=1,transform shape] 

\def\deltax{0.3} 
\def\deltay{0.5} 


\node [dot, fill=\groupStructColour] (mult) at (0,0) {};

\end{tikzpicture}}\!$-phases is $\integersMod{4}$ and the subgroup of $\hbox{\begin{tikzpicture} [scale=1,transform shape] 

\def\deltax{0.3} 
\def\deltay{0.5} 


\node [dot, fill=\groupStructColour] (mult) at (0,0) {};

\end{tikzpicture}}\!$-classical points is $\integersMod{2}$. Conclude with Theorem \ref{thm_MerminNonLocality} and Example \ref{example_PhaseGroupZX}.
	\end{proof}

	\begin{corollary} The toy theory Spekk from \cite{CQM-StrongComplementarity} is Mermin local.
	\end{corollary}
	\begin{proof} 
		Same setup as in the previous corollary, but the phase group is now $\integersMod{2} \times \integersMod{2}$. Conclude using Theorem \ref{thm_MerminLocality} and Example \ref{example_PhaseGroupRel} with $d=2$.
	\end{proof}
	
    \begin{corollary}
    Qutrit stabilizer quantum mechanics from~\cite{ranchin2014depicting} is Mermin local.
    \end{corollary}
    \begin{proof}
    The phase group here is $\mathbb{Z}_3\times\mathbb{Z}_3$. Conclude using Theorem~\ref{thm_MerminLocality} and Example~\ref{example_PhaseGroupRel} with $d=3$.\footnote{This example was first constructed by Edwards in~\cite{coecke2011phase} without reference to the qutrit stabilizer formalism.  This work also anticipated Example~\ref{example_PhaseGroupRel}, using a specific construction.}
    \end{proof}

	\begin{corollary} The category $\fRelCategory$ of finite sets and relations is Mermin local. 
	\end{corollary}
	\begin{proof} 
		See \cite{CQM-CQMnotes, StefanoGogioso-RepTheoryCQM} for more details on strong complementarity in $\fRelCategory$. Any $\dagger$-qSCFA on a set $\SpaceH$ in $\fRelCategory$ is a groupoid: we write it in the form $\oplus_{h\in H} G_h$, where $H$ is a set, $G_h$ are disjoint groups and $\cup_{h\in H} G_h = \SpaceH$. Any strongly complementary pair $\hbox{\begin{tikzpicture} [scale=1,transform shape] 

\def\deltax{0.3} 
\def\deltay{0.5} 


\node [dot, fill=\classicalStructColour] (mult) at (0,0) {};

\end{tikzpicture}}\!,\hbox{\begin{tikzpicture} [scale=1,transform shape] 

\def\deltax{0.3} 
\def\deltay{0.5} 


\node [dot, fill=\groupStructColour] (mult) at (0,0) {};

\end{tikzpicture}}\!$ is in the form $(\oplus_{h\in H} G,\oplus_{g\in G}H)$, where both $G$ and $H$ are groups (seen as sets when indexing the groupoids), and we can w.l.o.g. write $\SpaceH$ as $G \times H$. Each $\hbox{\begin{tikzpicture} [scale=1,transform shape] 

\def\deltax{0.3} 
\def\deltay{0.5} 


\node [dot, fill=\groupStructColour] (mult) at (0,0) {};

\end{tikzpicture}}\!$-classical points is in the form $\suchthat{(g,h)}{h\in H}$ for some $g\in G$, while the $\hbox{\begin{tikzpicture} [scale=1,transform shape] 

\def\deltax{0.3} 
\def\deltay{0.5} 


\node [dot, fill=\classicalStructColour] (mult) at (0,0) {};

\end{tikzpicture}}\!$-phases are in the form $\suchthat{(g_h,h)}{h\in H}$, for some family $(g_h)_{h \in H}$ of elements of $G$. Thus the group of $\hbox{\begin{tikzpicture} [scale=1,transform shape] 

\def\deltax{0.3} 
\def\deltay{0.5} 


\node [dot, fill=\classicalStructColour] (mult) at (0,0) {};

\end{tikzpicture}}\!$-phases is the group $G^H$ of $H$-indexed vectors with values in $G$, and the subgroup of $\hbox{\begin{tikzpicture} [scale=1,transform shape] 

\def\deltax{0.3} 
\def\deltay{0.5} 


\node [dot, fill=\groupStructColour] (mult) at (0,0) {};

\end{tikzpicture}}\!$-classical points, isomorphic to $G$, is that of vectors with constant components. Conclude using Theorem \ref{thm_MerminLocality} and Example \ref{example_PhaseGroupRel}.
	\end{proof}

\noindent This last result is particularly interesting for the following reasons:
    \begin{enumerate}
        \item[1.] Almost no $\dagger$-qSCFAs in $\fRelCategory$ have enough classical points (exactly one per space, out of a number that grows exponentially with space size).      
        \item[2.] The family of arguments from \cite{CQM-StrongComplementarity} fails in $\fRelCategory$ (partially as a consequence of the previous point).
        \item[3.] There are plenty of strongly complementary pairs in $\fRelCategory$, and arbitrarily many more $\hbox{\begin{tikzpicture} [scale=1,transform shape] 

\def\deltax{0.3} 
\def\deltay{0.5} 


\node [dot, fill=\classicalStructColour] (mult) at (0,0) {};

\end{tikzpicture}}\!$-phases than $\hbox{\begin{tikzpicture} [scale=1,transform shape] 

\def\deltax{0.3} 
\def\deltay{0.5} 


\node [dot, fill=\groupStructColour] (mult) at (0,0) {};

\end{tikzpicture}}\!$ classical points, but the lack of \textit{algebraically non-trivial} phases results in $\fRelCategory$ being Mermin local.
        \item[4.] As a consequence of point 3, quantum protocols relying only on Mermin non-locality will show no quantum advantage in $\fRelCategory$.
    \end{enumerate}


\section{Mermin in $\fdHilbCategory$: beyond the complementary $XY$ pair}
\label{section:non-compl}

We now focus on $\fdHilbCategory$ and quantum mechanics. While in general we can have many different choices of measurement on each subsystem (see Definition \ref{def_MerminMeasurements}), we shall restrict to the case of only two distinct measurements, i.e. $(N,M=2,D)$ scenarios.  In the case of qubits and $(N,2,2)$ scenarios, these complementary measurements happen to be the only choices that will lead to a non-locality argument. One might then conjecture that this will be the case for any dimension.  In this section we show that this assumption is not the case. For $(N,2,D)$ scenarios it is not necessary to have the two measurements be complementary. There are many possible pairs in general.

\begin{definition}
A {\bf two-measurement Mermin scenario} for $N$ systems (each with $D$ dimensions) and strongly complementary GHZ observable with $\hbox{\begin{tikzpicture} [scale=1,transform shape] 

\def\deltax{0.3} 
\def\deltay{0.5} 


\node [dot, fill=\classicalStructColour] (mult) at (0,0) {};

\end{tikzpicture}}\!$-phase group $G$ is denoted $G(N,2,D)$. Each system has two possible measurement settings:
\begin{enumerate}
\item the first measurement observable is the D-dimensional $X$ observable,
\item and the second measurement observable $B$ is defined by a Z-phase gate applied to $X$.
\end{enumerate}
In  general, the form of $B$ can be specified by the $D$-dimensional Z-phase applied to $X$. This Z-phase is of the form $(1, e^{1b_1}, ..., e^{ib_{D-1}})^{T}$ with $D-1$ degrees of freedom. A two-measurement Mermin scenario thus consists of $V$ variations each with $\beta$ measurements of the $B$ observable.
\end{definition}

\begin{example}
For qubits there is only a single possible phase group: $\mathbb{Z}_2$. A Mermin argument for three qubits (denoted $\mathbb{Z}_2(3,2,2))$ has measurements of the usual $X$ observable and of the $B$ observable that is a phase applied to $X$, i.e. diag$(1,e^{ib_1})X$. In the traditional Mermin scenario $\mathbb{Z}_2(3,2,2)$ from \cite{mermin1990quantum}, we have $V = 3$ and $\beta = 2$.
\end{example}

\noindent The state presented in Diagram \ref{MerminSetupProofLHVeval} will be zero when the \emph{control point} on the left is distinct from the \emph{variations point} on the right. We can characterize this as a condition on $B$ in our two measurement scenario with the following theorem.

\begin{lemma}
\label{correspond}
Measurements $X$ and $B$ allow a $(N,2,D)$ Mermin non-locality argument iff 
\begin{align}
\label{NewCond}
\sum_{j=1}^{D-1}e^{ic_j} = -1, \qquad \mbox{where } c_j = b_j\left(\bigoplus_{i=1}^{V}\beta\right) ,
\end{align}
where the sum in $c_j$ is the group sum for the $\hbox{\begin{tikzpicture} [scale=1,transform shape] 

\def\deltax{0.3} 
\def\deltay{0.5} 


\node [dot, fill=\classicalStructColour] (mult) at (0,0) {};

\end{tikzpicture}}\!$-phase group $G$.
\end{lemma}
\begin{proof}
Diagram \ref{MerminSetupProofLHVeval} implies that the Mermin argument will succeed when the control point and variations point are distinct classical points.  In $\fdHilbCategory$ this precisely means that they are orthogonal vectors.  The vector that represents the control point is given by the $D$-dimensional unit for the $X$ observable, i.e. $1/\sqrt{D}(1,1,...,1)^{T}$.
The variations point is then given by the group sum of other classical points specified by their phase.  The phase for each classical point is given by the sum of phase accumulated by each $B$ measurement.  As there are $\beta$ such measurements in each variation, their sum is given by
\[
\frac{1}{\sqrt{D}}\left(\begin{array}{c}1 \\ e^{i \beta b_1} \\\vdots \\ e^{i \beta b_{D-1}} \end{array}\right)_1\oplus\left(\begin{array}{c}1 \\ e^{i \beta b_1} \\\vdots \\ e^{i \beta b_{D-1}} \end{array}\right)_2\oplus...\oplus\left(\begin{array}{c}1 \\ e^{i \beta b_1} \\\vdots \\ e^{i \beta b_{D-1}} \end{array}\right)_V = \frac{1}{\sqrt{D}}\left(\begin{array}{c}1 \\ e^{ic_1} \\\vdots \\ e^{ic_{D-1}} \end{array}\right),
\]
where the constants $c_j$ are defined as in Equation \ref{NewCond}.
Orthogonality between the control and variations points then requires
\begin{align*}
\left(\begin{array}{cccc}1 & 1 & ... & 1 \end{array}\right)
\left(\begin{array}{c}1 \\ e^{ic_1} \\\vdots \\ e^{ic_{D-1}} \end{array}\right) &= 0 \qquad \Rightarrow \qquad
\sum_{j=1}^{D-1}e^{ic_j} = -1
\end{align*}
This exactly recovers Equation \ref{NewCond} and completes the proof.
\end{proof}

\noindent This gives a necessary and sufficient condition on these measurements to enable a Mermin non-locality test. Note that in Mermin's original scenario  measurement observables were necessarily complementary, but that in general this is not the case.
\begin{theorem}
\label{thm:qubitcase}
In $(3,2,2)$ three qubit Mermin scenarios, the two measurements must be complementary.
\end{theorem}
\begin{proof}
We have $V=3$, $\beta = 2$, $G=\mathbb{Z}_2$ and $D = 2$. Thus
\begin{align*}
c_j &= \beta b_j \left(\bigoplus^3_{l=1}1\right)= 2b_j (3 \mod 2) = 2b_j
\end{align*}
so that our condition on $B$ becomes
\begin{align*}
\sum_{j=1}^{D-1}e^{ic_j} &= e^{i2b_1} = -1 \Rightarrow
b_1 = \frac{\pi}{2}
\end{align*}
with only a single solution.  This means that in this scenario there is only one measurement that could be used with $X$.  This is the $Y$ observable and it is complementary to $X$.
\end{proof}

\begin{theorem}
\label{thm:noncompl}
For $(N,2,D)$ scenarios the measurements need not be complementary.
\end{theorem}
\begin{proof}
We prove this by counterexample.  Consider the three dimensional  ($D=3$) five party Mermin scenario.  The phase group of the non-local state is then given by $G=\mathbb{Z}_3$.
The control measurement is given by five systems all measured by the $X$ observable, i.e. $XXXXX$. The variations are
\begin{align*}
BBBXX \qquad BBXBX \qquad BXBBX \qquad XBBBX \qquad XBXBB \\
BBXXB \qquad BXBXB \qquad XBBXB \qquad BXXBB \qquad XXBBB
\end{align*}
so that $V = 10$ and $\beta = 3$. We calculate the coefficients
\begin{align*} c_j &= \beta b_j \left(\bigoplus^{10}_{l=1}1\right)= 3b_j (10\mod 3) = 3b_j
\end{align*}
Observable $B$ must then satisfy $e^{i3b_1}+e^{i3b_2} = -1$. Any $B$ observable satisfies this condition if  $b_2 = -\frac{i}{3}\log\left[-1-e^{3ib_1}\right]$.
Consider $b_1 = \frac{2\pi}{9} \Rightarrow b_2 =- \frac{2\pi}{9}$ and calculate (for $\omega = e^{2\pi i/3}$):
\begin{align*}
B :: \left(\begin{array}{ccc}1 & 0 & 0 \\ 0 & e^{i2\pi/9} & 0 \\ 0 & 0 & e^{-i2\pi/9} \end{array}\right) \frac{1}{\sqrt{3}}
\left(\begin{array}{ccc}1 & 1 & 1 \\ 1 & \omega & \omega^2 \\ 1 & \omega^2 & \omega^4  \end{array}\right) 
= \frac{1}{\sqrt{3}}\left(\begin{array}{ccc}1 & 1 & 1 \\ e^{2i\pi/9} & e^{8i\pi/9} & e^{-4i\pi/9} \\ e^{-2i\pi/9} & e^{-8i\pi/9} & e^{4i\pi/4} \end{array}\right)
\end{align*}
Observable $B$ is clearly not complementary to $X$ by simply checking the dot products of their basis vectors.
\end{proof}

\noindent Further we can exhibit numerical results that calculate the number of Mermin effective measurement pairs available for a particular scenario. For a given number of parties $N$ we have calculated the number of effective pairs maximized over all viable variation choices.  Typically these maximum values are found for variations where $\beta$ is maximized. Figure 1a shows pair counts for $\mathbb{Z}_2(N,2,2)$ scenarios. Here is appears that the number of effective measurement pairings grows approximately linearly with the number of parties. Figure 1b shows pair counts for the more complex $\mathbb{Z}_3(N,2,3)$ scenarios. It is clear that there are many (in some cases thousands) more available measurement configurations than just those given by complementary observables. This vastly expands the number of experimental setups that will generate, with certainty, a non-locality violation.  Indeed, combining this result with those of Section \ref{section_QSS} opens up a large class of quantum secret sharing protocols based on non-complementary measurements.

\begin{figure}[th]
    \label{counts}
    \centering
    \begin{subfigure}[t]{0.5\textwidth}
        \centering
        \includegraphics[height=1.6in]{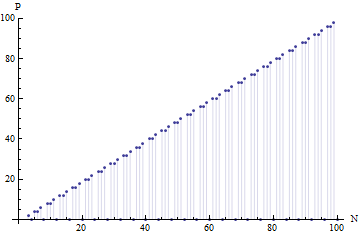}
        \caption{ }
    \end{subfigure}%
    ~ 
    \begin{subfigure}[t]{0.5\textwidth}
        \centering
        \includegraphics[height=1.6in]{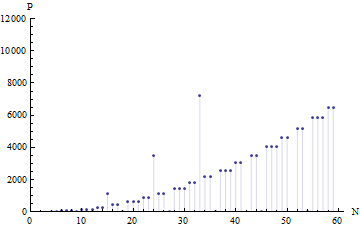}
        \caption{ }
    \end{subfigure}
    \caption{(a) A plot of the number of Mermin effective measurement pairs $P$ vs. the number of parties in the Mermin scenario $N$ for $\mathbb{Z}_2(N,2,2)$ scenarios. (b) A plot of the number of effective pairs for $\mathbb{Z}_3(N,2,3)$ scenarios.  These numbers were obtained by numerically counting solutions to~\eqref{NewCond}.}
\end{figure}

\section{Quantum Secret Sharing: non-locality as a resource}
\label{section_QSS}

The HBB CQ (N,N) family of Quantum Secret Sharing protocols originates in \cite{HBB, HBB2}, and has been abstractly formulated in Categorical Quantum Mechanics \cite{HBBVlad}. Here we generalise their construction to abstract process theories, unearthing a deep connection with Mermin non-locality. 

This protocol requires a pair $(\hbox{\begin{tikzpicture} [scale=1,transform shape] 

\def\deltax{0.3} 
\def\deltay{0.5} 


\node [dot, fill=\classicalStructColour] (mult) at (0,0) {};

\end{tikzpicture}}\!,\hbox{\begin{tikzpicture} [scale=1,transform shape] 

\def\deltax{0.3} 
\def\deltay{0.5} 


\node [dot, fill=\groupStructColour] (mult) at (0,0) {};

\end{tikzpicture}}\!)$ of strongly complementary observables, and an $(N+1)$-partite GHZ state shared by the dealer and the $N$ players. The dealer (and nobody else) knows the (classical) secret, in the form of a $\hbox{\begin{tikzpicture} [scale=1,transform shape] 

\def\deltax{0.3} 
\def\deltay{0.5} 


\node [dot, fill=\groupStructColour] (mult) at (0,0) {};

\end{tikzpicture}}\!$-classical point. The aim of the protocol is for the dealer to broadcast some information to all players on a public classical channel, and for the secret to be deterministically decodeable if only if all $N$ players cooperate. The implementation, graphically summarised in \ref{HBBQSS}, goes as follows:
\begin{enumerate}
\item[1.] the dealer and the players agree on a random set of $\hbox{\begin{tikzpicture} [scale=1,transform shape] 

\def\deltax{0.3} 
\def\deltay{0.5} 


\node [dot, fill=\classicalStructColour] (mult) at (0,0) {};

\end{tikzpicture}}\!$-phases $\alpha_0,\alpha_1,...,\alpha_N$ such that $\sum \alpha_j $ is some $\hbox{\begin{tikzpicture} [scale=1,transform shape] 

\def\deltax{0.3} 
\def\deltay{0.5} 


\node [dot, fill=\groupStructColour] (mult) at (0,0) {};

\end{tikzpicture}}\!$-classical point (call it $a$). This operation is done on a public channel.
\item[2.] the dealer measures his part of the system of the system with phase $\alpha_0$, and uses the resulting $\hbox{\begin{tikzpicture} [scale=1,transform shape] 

\def\deltax{0.3} 
\def\deltay{0.5} 


\node [dot, fill=\groupStructColour] (mult) at (0,0) {};

\end{tikzpicture}}\!$-classical data to encode the plaintext secret (classically adding the secret and the measurement data in the group $K_{\;\hbox{\begin{tikzpicture} [scale=1,transform shape] 

\def\deltax{0.3} 
\def\deltay{0.5} 


\node [dot, fill=\groupStructColour] (mult) at (0,0) {};

\end{tikzpicture}}\!}$; this generalises the original XOR operation, corresponding to $K_{\;\hbox{\begin{tikzpicture} [scale=1,transform shape] 

\def\deltax{0.3} 
\def\deltay{0.5} 


\node [dot, fill=\groupStructColour] (mult) at (0,0) {};

\end{tikzpicture}}\!} = \integersMod{2}$ with addition mod 2) into a classical cyphertext. This operation is done locally and privately by the dealer.
\item[3.] the dealer broadcasts the cyphertext on a public classical channel to the players.
\item[4.] at some later stage, when they all agree to unveil the secret, the $N$ players measure their part of the system, each locally and privately.
\item[5.] all players broadcast the $\hbox{\begin{tikzpicture} [scale=1,transform shape] 

\def\deltax{0.3} 
\def\deltay{0.5} 


\node [dot, fill=\groupStructColour] (mult) at (0,0) {};

\end{tikzpicture}}\!$-classical results of their measurements on a public classical channel.
\item[6.] the broadcast results can be classically added in $K_{\; \hbox{\begin{tikzpicture} [scale=1,transform shape] 

\def\deltax{0.3} 
\def\deltay{0.5} 


\node [dot, fill=\groupStructColour] (mult) at (0,0) {};

\end{tikzpicture}}\!}$, then the result can be added to $a$ and finally to the cyphertext (again in the group $K_{\; \hbox{\begin{tikzpicture} [scale=1,transform shape] 

\def\deltax{0.3} 
\def\deltay{0.5} 


\node [dot, fill=\groupStructColour] (mult) at (0,0) {};

\end{tikzpicture}}\!}$) to recover the original $\hbox{\begin{tikzpicture} [scale=1,transform shape] 

\def\deltax{0.3} 
\def\deltay{0.5} 


\node [dot, fill=\groupStructColour] (mult) at (0,0) {};

\end{tikzpicture}}\!$-classical plaintext secret. 
\end{enumerate}
\begin{equation}\label{HBBQSS}
	\hbox{\begin{tikzpicture}[node distance = 10mm]

\node (a) {};
\node (topdot) [smalldot, below of = a, fill = \Zcolour, yshift = 3mm] {};
\node (antipode) [below left of = topdot,  yshift = 2mm,xshift = -3mm] {};
\node (antipodefiller) [dot, inner sep = 1mm,fill = \Zcolour, below right of = topdot, yshift = 2mm,xshift = +6mm] {$a$};

\node (ldot) [smalldot, below left of = antipode, fill = \Zcolour, yshift = -10mm, xshift = -3mm] {};
\node (rdot) [smalldot, below right of = antipodefiller, fill = \Zcolour, yshift = 2mm, xshift = + 6mm] {};

\node (lcontroldot) [below left of = ldot, yshift = -15mm, xshift = -5mm] {secret};

\node (rcontroldot) [smalldot, below right of = ldot, fill = \Xcolour, yshift = -2mm,xshift = +2mm] {};
\node (alpha0) [smalldot,inner sep = 0.1mm, below right of = rcontroldot, fill = \Zcolour,xshift = -3mm] {$+\alpha_0$};
\node (alpha0star) [smalldot,inner sep = 0.1mm, below left of = rcontroldot, fill = \Zcolour,xshift = +3mm] {$-\alpha_0$};

\node (lvariationdot) [smalldot, below left of = rdot, fill = \Xcolour, yshift = 0mm,xshift = -8mm] {};
\node (alpha1) [smalldot,inner sep = 0.1mm, below right of = lvariationdot, fill = \Zcolour,xshift = -3mm] {$+\alpha_1$};
\node (alpha1star) [smalldot,inner sep = 0.1mm, below left of = lvariationdot, fill = \Zcolour,xshift = +3mm] {$-\alpha_1$};

\node (rvariationdot) [smalldot, below right of = rdot, fill = \Xcolour, yshift = 0mm,xshift = +8mm] {};
\node (alphaN) [smalldot,inner sep = 0.1mm, below right of = rvariationdot, fill = \Zcolour,xshift = -3mm] {$+\alpha_N$};
\node (alphaNstar) [smalldot,inner sep = 0.1mm, below left of = rvariationdot, fill = \Zcolour,xshift = +3mm] {$-\alpha_N$};

\node (lGHZ) [smalldot,fill=\Zcolour,below of = lvariationdot, yshift = -27mm] {}; 
\node (rGHZ) [smalldot,fill=\Zcolour,right of = lGHZ,yshift = 2mm] {}; 

\begin{pgfonlayer}{background}
\draw[->-=.5,out=90,in=270] (topdot) to (a);
\draw[->-=.5,out=45,in=225] (ldot) to (topdot);
\draw[->-=.5,out=135,in=315] (rdot) to (antipodefiller);
\draw[->-=.5,out=135,in=315] (antipodefiller) to (topdot);

\draw[->-=.5,out=90,in=225] (lcontroldot) to (ldot);
\draw[->-=.5,out=135,in=315] (rcontroldot) to (ldot);
\draw[->-=.5,out=45,in=225] (lvariationdot) to (rdot);
\draw[->-=.5,out=135,in=315] (rvariationdot) to (rdot);

\draw[-<-=.5,out=90,in=225] (alpha0star) to (rcontroldot);
\draw[->-=.5,out=90,in=315] (alpha0) to (rcontroldot);

\draw[-<-=.5,out=90,in=225] (alpha1star) to (lvariationdot);
\draw[->-=.5,out=90,in=315] (alpha1) to (lvariationdot);

\draw[-<-=.5,out=90,in=225] (alphaNstar) to (rvariationdot);
\draw[->-=.5,out=90,in=315] (alphaN) to (rvariationdot);

\draw[->-=.5,out=45,in=270] (rGHZ) to (alphaN);
\draw[->-=.5,out=90,in=270] (rGHZ) to (alpha1);
\draw[->-=.5,out=135,in=270] (rGHZ) to (alpha0);

\draw[->-=.5,out=270,in=45] (alphaNstar) to (lGHZ);
\draw[->-=.5,out=270,in=90] (alpha1star) to (lGHZ);
\draw[->-=.5,out=270,in=135] (alpha0star) to (lGHZ);

\end{pgfonlayer}

\node (lend) [below of = alpha0star, yshift = -7mm] {window of attack};
\node (rend) [below of = alphaN, yshift = -5mm, xshift = 5mm] {};

\begin{pgfonlayer}{background}
\draw[dashed,out = 45, in=135] (lend) to (rend);
\end{pgfonlayer}

\end{tikzpicture}}
\end{equation}

\noindent Most of the operations are either done locally and privately (all the measurements and the secret encoding), or broadcast by design on public classical channels, where one assumes that integrity of the message is guaranteed by appropriate classical protocols. There are many additional layers of quantum guarantees coming with this protocol, depending on the level of tampering allowed and on the phases chosen:
\begin{enumerate}
\item[1.] Assume no tampering happens anywhere. Then the refusal of (at least) one player to broadcast his or her measurement result makes the secret totally random to anyone else.
\item[2.] Assume that an attacker is allowed to tamper \textit{only} with the GHZ state, and \textit{before} the phases are chosen. Then the maximum amount of information she can gain is limited by (a) the random distribution on phases and (b) the amount of bias between the possible phases for each system. If $p_{max}$ is the highest probability appearing in the distribution of the phase choices (traditionally uniform with probability $1/^N$)\footnote{Not $1/2^{N+1}$, because of the parity requirement.}, and we let $k := |K_{\;\hbox{\begin{tikzpicture} [scale=1,transform shape] 

\def\deltax{0.3} 
\def\deltay{0.5} 


\node [dot, fill=\groupStructColour] (mult) at (0,0) {};

\end{tikzpicture}}\!}|$ be the dimensionality of the space (traditionally $k=2$ for qubits), then optimal tampering reveals an average of $p_{max}$ $k$-its of classical information (on a secret of 1 $k$-it), in the case where the alternative measurements on each system are mutually unbiased (e.g. the traditional $X,Y$ pair). A more complicated failure expression can be worked out for arbitrary bases. This gain in information, however, is compensated by the introduction of a probability of failure for the entire protocol of $(1-p_{max}) \cdot (1-1/k)$ (again in the mutually unbiased case), which can be detected by the players/dealer via statistical analysis of the outcomes.
\item[3.] The kind of tampering allowed in the previous point does not give significant advantage to the attacker (at least for large number of players), and can be mitigated by appropriate statistical analysis of the measurement outputs; however, there is a stronger form of tampering that we can consider. Assume that the attacker is allowed to tamper with the GHZ state after the phases have been chosen, or even with the measurement devices of the dealer/player themselves, in a way that will ensure he knows the measurement outcomes with certainty beforehand; this is the model of attack assumed by device-independent security, pioneered in \cite{QTC-NoSignallingQKD}. Under this stronger model of attack, we can show that the protocol is secure if and only if the phases chosen by the players are algebraically non-trivial. Indeed, from the point of view of the dealer/players, the attack results in the measurement outcomes having a classical probability distribution: 
\begin{enumerate}
\item[(a)] if the phases are algebraically non-trivial, the probability distribution in the tampered case will never match, because of contextuality, that generated by the un-tampered protocol, and the attack can be detected by statistical analysis of the outcomes.
\item[(b)] if the phases are algebraically trivial, on the other hand, they admit a probabilistic local hidden variable, and the attacker can generate her deterministic outcomes in a way to mimic the probability distribution of the un-tampered protocol.
\end{enumerate}
\end{enumerate}

\noindent To summarise, there are three distinct quantum resources playing complementary roles in the security of this protocol: the entanglement structure of the GHZ state, the amount of mutual complementarity of the available phases, and their algebraic non-triviality. Firstly, the entanglement structure of the GHZ state is the resource ensuring that the refusal of one player to cooperate results, if no tampering is allowed, into the inability for everyone else to recover the secret. Secondly, the amount of mutual complementarity of the available phases, e.g. the complementarity of the $X,Y$ pair, limits the maximum amount of information an attacker can gain by tampering with the state before phases are chose, and the minimum amount of disturbance introduced by the attack. Finally, Mermin non-locality, or equivalently algebraic non-triviality of the chosen phases, is the key resource ensuring device-independent security of the protocol.



\section{Conclusions and future work}
	\label{section_conclusion}
	By using few, simple ingredients --- $\dagger$-SMCs, strongly complementary pairs, GHZ states, phases and classical points --- we have generalised Mermin measurements to arbitrary abstract process theories. 
	We have defined Mermin non-locality, and we have proven that a necessary and sufficient\footnote{Always necessary, sufficient under the assumption that classical points form a basis.} condition for it is the existence of algebraically non-trivial phases, i.e. of phases which satisfy equations that classical points cannot. 
	As a corollary, we have confirmed the well-known result that the stabilizer ZX calculus (and therefore $\fdHilbCategory$) is Mermin non-local, and we have proven that $\fRelCategory$, a toy category of choice for Categorical Quantum Mechanics, is Mermin local (despite its unboundedly large ratio of phases to classical points). 
	This characterisation as the existence of certain phases opens the way to the treatment of Mermin non-locality as a resource in the abstract design of quantum protocols, as we have exemplified with the HBB CQ family of Quantum Secret Sharing protocols. 
	Finally, the application of our general framework to Mermin-type experiment in quantum mechanics allows us to show that, even in the restricted case of two-measurement scenarios, complementary measurements are not necessary, leading to many more potential configurations than previously believed.
	We conclude with a few open questions for investigation:
        \begin{enumerate}
        \item What are the minimal conditions under which algebraically non-trivial phases lead to non-locality?
        \item What is the exact connection between this framework as the framework of Abramsky et al.~\cite{NLC-AvN} for generalised All-versus-Nothing arguments where measurement outcomes are elements of some general field?
        \item Is there a more informative group-theoretic formulation of the algebraic non-triviality used here?
        \item Our analysis focuses on non-locality paradoxes for a kind of GHZ state.  It was recently shown by~\cite{tang2013greenberger} that multipartite non-locality arguments can be constructed from any of a set of qudit graph states that they call GHZ graphs.  What are the connections between these qudit graph states and the phase group formalism we present here?
        \item Which other quantum algorithms depend on Mermin non-locality as a resource to transcend classicality? Which process theories show these characteristics?
        \end{enumerate}

\paragraph{Acknowledgements}
	The authors would like to thank Bob Coecke and Aleks Kissinger for plenty of comments and suggestions, as well as Sukrita Chatterji and Nicol\`o Chiappori for useful discussions and support. Funding from EPSRC (award number OUCL/2013/SG) and Trinity College for the first author and the Rhodes Trust and AFOSR grant FA9550-14-1-0079 for the second is gratefully acknowledged. Both authors contributed equally to this work.
	
\bibliographystyle{eptcs}
\bibliography{bibliography/CategoryTheory,bibliography/CategoricalQM,bibliography/NonLocalityContextuality,bibliography/QuantumComputing,bibliography/ClassicalMechanics,bibliography/LogicComputation,bibliography/Gravitation,bibliography/QFT,bibliography/StatisticalPhysics,bibliography/Misc,bibliography/StefanoGogioso,bibliography/ComputationalLinguistics,bibliography/MerminNonLocal}

\appendix

\section{Preliminary definitions}
\label{app:defs}

In this section we recall some basic background definitions from the literature~\cite{CQM-ZXCalculusSeminal}.  We will use the diagrammatic language of symmetric monoidal categories, c.f.~\cite{selinger2011survey}.

\begin{definition}
\label{def:frobenius}
In a $\dagger$-symmetric monoidal category ($\dagger$-SMC), the pair of a monoid $(A,\hbox{\begin{tikzpicture} [scale=1,transform shape] 

\def\deltax{0.3} 
\def\deltay{0.5} 


\node (mult_label_inl) at (-\deltax,-\deltay) {};
\node (mult_label_inr) at (+\deltax,-\deltay) {};
\node [dot, fill=\groupStructColour] (mult) at (0,0) {};
\node (mult_label_out) at (0,+\deltay) {};

\draw[-] [out=90,in=225](mult_label_inl) to (mult);
\draw[-] [out=90,in=315](mult_label_inr) to (mult);
\draw[-] (mult) to (mult_label_out);

\end{tikzpicture}}\!,\hbox{\begin{tikzpicture} [scale=1,transform shape] 

\def\deltax{0.3} 
\def\deltay{0.5} 


\node [dot, fill=\groupStructColour] (mult) at (0,0) {};
\node (mult_label_out) at (0,+\deltay) {};
\draw[-] (mult) to (mult_label_out);

\end{tikzpicture}}\!)$ and comonoid $(A,\hbox{\begin{tikzpicture} [scale=1,transform shape] 

\def\deltax{0.3} 
\def\deltay{0.5} 


\node (mult_label_outl) at (-\deltax,+\deltay) {};
\node (mult_label_outr) at (+\deltax,+\deltay) {};
\node [dot, fill=\groupStructColour] (mult) at (0,0) {};
\node (mult_label_in) at (0,-\deltay) {};
\draw[-] [in=270,out=135] (mult) to (mult_label_outl);
\draw[-] [in=270,out=45] (mult) to (mult_label_outr);
\draw[-] (mult_label_in) to (mult);

\end{tikzpicture}}\!,\hbox{\begin{tikzpicture} [scale=1,transform shape] 

\def\deltax{0.3} 
\def\deltay{0.5} 


\node [dot, fill=\groupStructColour] (mult) at (0,0) {};
\node (mult_label_in) at (0,-\deltay) {};
\draw[-] (mult_label_in) to (mult);

\end{tikzpicture}}\!)$ form a \textbf{dagger-Frobenius algebra} (or $\dagger$-FA) when the following equation holds:
\begin{equation}
\label{eq:frobenius}
\begin{aligned}
\begin{tikzpicture}[xscale=1, yscale=1]
\node (a) [smalldot, fill = \groupStructColour] at (0.5,2) {};
\draw (0,0) to (0,1) to [out=up, in=\swangle] (a);
\draw (a) to (0.5,3);
\node (b) [smalldot, fill = \groupStructColour] at (1.5,1) {};
\draw (a) to [out=\seangle, in=\nwangle] (b);
\draw (1.5,0) to (b) to [out=\neangle, in=down] (2,2) to (2,3);
\end{tikzpicture}
\end{aligned}
\quad = \quad
\begin{aligned}
\begin{tikzpicture}[xscale=1, yscale=1]
\node (a) [smalldot, fill = \groupStructColour] at (-0.5,2) {};
\draw (0,0) to (0,1) to [out=up, in=\seangle] (a);
\draw (a) to (-0.5,3);
\node (b) [smalldot, fill = \groupStructColour] at (-1.5,1) {};
\draw (a) to [out=\swangle, in=\neangle] (b);
\draw (-1.5,0) to (b) to [out=\nwangle, in=down] (-2,2) to (-2,3);
\end{tikzpicture}
\end{aligned}
\end{equation}
\end{definition}

	\begin{definition}\label{def_QuasiSpecial}
		A \textbf{quasi-special} $\dagger$-Frobenius algebra  $(\hbox{\begin{tikzpicture} [scale=1,transform shape] 

\def\deltax{0.3} 
\def\deltay{0.5} 


\node (mult_label_inl) at (-\deltax,-\deltay) {};
\node (mult_label_inr) at (+\deltax,-\deltay) {};
\node [dot, fill=\groupStructColour] (mult) at (0,0) {};
\node (mult_label_out) at (0,+\deltay) {};

\draw[-] [out=90,in=225](mult_label_inl) to (mult);
\draw[-] [out=90,in=315](mult_label_inr) to (mult);
\draw[-] (mult) to (mult_label_out);

\end{tikzpicture}}\!,\hbox{\begin{tikzpicture} [scale=1,transform shape] 

\def\deltax{0.3} 
\def\deltay{0.5} 


\node [dot, fill=\groupStructColour] (mult) at (0,0) {};
\node (mult_label_out) at (0,+\deltay) {};
\draw[-] (mult) to (mult_label_out);

\end{tikzpicture}}\!,\hbox{\begin{tikzpicture} [scale=1,transform shape] 

\def\deltax{0.3} 
\def\deltay{0.5} 


\node (mult_label_outl) at (-\deltax,+\deltay) {};
\node (mult_label_outr) at (+\deltax,+\deltay) {};
\node [dot, fill=\groupStructColour] (mult) at (0,0) {};
\node (mult_label_in) at (0,-\deltay) {};
\draw[-] [in=270,out=135] (mult) to (mult_label_outl);
\draw[-] [in=270,out=45] (mult) to (mult_label_outr);
\draw[-] (mult_label_in) to (mult);

\end{tikzpicture}}\!,\hbox{\begin{tikzpicture} [scale=1,transform shape] 

\def\deltax{0.3} 
\def\deltay{0.5} 


\node [dot, fill=\groupStructColour] (mult) at (0,0) {};
\node (mult_label_in) at (0,-\deltay) {};
\draw[-] (mult_label_in) to (mult);

\end{tikzpicture}}\!)$ in a $\dagger$-SMC is a Frobenius algebra that satisfies:
		\begin{equation}\label{eqn_QuasiSpecialDef}
			\hbox{\begin{tikzpicture}[node distance=8mm]

\node (center) {};

\node (comult) [smalldot, fill=\groupStructColour] [below of = center, yshift = 5mm] {};
\node (mult) [smalldot, fill=\groupStructColour] [above of = center, yshift = -5mm] {};
\node (out) [above of = center, yshift = +1mm] {};
\node (in) [below of = center, yshift = -1mm] {};

\begin{pgfonlayer}{background}
\draw[-,out=45,in=315] (comult) to (mult);
\draw[-,out=135,in=225] (comult) to (mult);
\draw[-,out=90,in=270] (in) to (comult);
\draw[-,out=90,in=270] (mult) to (out);
\end{pgfonlayer}

\node (equals) [right of = center, xshift = 0mm] {$=$};
\node (center) [right of = equals, xshift = 5mm] {};
\node (out) [above of = center, yshift = +1mm] {};
\node (in) [below of = center, yshift = -1mm] {};
\node (scalar) [scalar] [left of = center, xshift = 3mm] {$N$};

\begin{pgfonlayer}{background}
\draw[-,out=90,in=270] (in) to (out);
\end{pgfonlayer}

\end{tikzpicture}}
		\end{equation}
		for some invertible scalar $N$.
	\end{definition}

These $\dagger$-qsFA are \textbf{commutative} when the monoid and comonid are commutative.
\begin{theorem}[{\cite[Thm 5.1]{CQM-OrthogonalBases}}]
Commutative dagger Frobenius algebras in $\fdHilbCategory$ are orthogonal bases.
\end{theorem}
\noindent The additional condition of specialness (quasi-specialness where $N=1$) for $\dagger$-qsCFA acts as a normalizing condition so that:
\begin{theorem}[{\cite[Sec 6]{CQM-OrthogonalBases}}]
\label{thm:cstructBases}
Commutative dagger Frobenius algebras in $\fdHilbCategory$ in $\fdHilbCategory$ are orthonormal bases.
\end{theorem}

\begin{definition}
\label{def:copyables}
The set of \textbf{classical states} $K_{\hbox{\begin{tikzpicture} [scale=1,transform shape] 

\def\deltax{0.3} 
\def\deltay{0.5} 


\node [dot, fill=\groupStructColour] (mult) at (0,0) {};

\end{tikzpicture}}\!}$ for a $\dagger$-Frobenius algebra  $(A, \hbox{\begin{tikzpicture} [scale=1,transform shape] 

\def\deltax{0.3} 
\def\deltay{0.5} 


\node (mult_label_inl) at (-\deltax,-\deltay) {};
\node (mult_label_inr) at (+\deltax,-\deltay) {};
\node [dot, fill=\groupStructColour] (mult) at (0,0) {};
\node (mult_label_out) at (0,+\deltay) {};

\draw[-] [out=90,in=225](mult_label_inl) to (mult);
\draw[-] [out=90,in=315](mult_label_inr) to (mult);
\draw[-] (mult) to (mult_label_out);

\end{tikzpicture}}\!,\hbox{\begin{tikzpicture} [scale=1,transform shape] 

\def\deltax{0.3} 
\def\deltay{0.5} 


\node [dot, fill=\groupStructColour] (mult) at (0,0) {};
\node (mult_label_out) at (0,+\deltay) {};
\draw[-] (mult) to (mult_label_out);

\end{tikzpicture}}\!,\hbox{\begin{tikzpicture} [scale=1,transform shape] 

\def\deltax{0.3} 
\def\deltay{0.5} 


\node (mult_label_outl) at (-\deltax,+\deltay) {};
\node (mult_label_outr) at (+\deltax,+\deltay) {};
\node [dot, fill=\groupStructColour] (mult) at (0,0) {};
\node (mult_label_in) at (0,-\deltay) {};
\draw[-] [in=270,out=135] (mult) to (mult_label_outl);
\draw[-] [in=270,out=45] (mult) to (mult_label_outr);
\draw[-] (mult_label_in) to (mult);

\end{tikzpicture}}\!,\hbox{\begin{tikzpicture} [scale=1,transform shape] 

\def\deltax{0.3} 
\def\deltay{0.5} 


\node [dot, fill=\groupStructColour] (mult) at (0,0) {};
\node (mult_label_in) at (0,-\deltay) {};
\draw[-] (mult_label_in) to (mult);

\end{tikzpicture}}\!)$ are all states $j:I\to A$ such that:
\begin{equation}
\label{eq:copy}
\begin{aligned}
\begin{tikzpicture}[yscale=-1]
\node [kpoint] (s) at (1.25,3) {$j$};
\draw (2,0) to [out=up, in=\seangle] (1.25,1.5);
\draw (0.5,0) to [out=up, in=\swangle] (1.25, 1.5);
\draw (1.25,1.5) to (s);
\node [smalldot, fill = \groupStructColour] at (1.25,1.5) {};
\end{tikzpicture}
\end{aligned}
\quad=\quad
\begin{aligned}
\begin{tikzpicture}
\node (a) [kpoint] at (2.5,0) {$j$};
\node (b) [kpoint] at (0,0) {$j$};
\draw (a) to (2.5,2);
\draw (b) to (0,2);
\end{tikzpicture}
\end{aligned}
\end{equation}
\end{definition}

	    We now define strong complementarity, the first fundamental ingredient of Mermin measurements.
	\begin{definition}\label{def_StrongComplementarity}
		A pair of $\dagger$-qSFAs $(\hbox{\begin{tikzpicture} [scale=1,transform shape] 

\def\deltax{0.3} 
\def\deltay{0.5} 


\node (mult_label_inl) at (-\deltax,-\deltay) {};
\node (mult_label_inr) at (+\deltax,-\deltay) {};
\node [dot, fill=\groupStructColour] (mult) at (0,0) {};
\node (mult_label_out) at (0,+\deltay) {};

\draw[-] [out=90,in=225](mult_label_inl) to (mult);
\draw[-] [out=90,in=315](mult_label_inr) to (mult);
\draw[-] (mult) to (mult_label_out);

\end{tikzpicture}}\!, \hbox{\begin{tikzpicture} [scale=1,transform shape] 

\def\deltax{0.3} 
\def\deltay{0.5} 


\node [dot, fill=\groupStructColour] (mult) at (0,0) {};
\node (mult_label_out) at (0,+\deltay) {};
\draw[-] (mult) to (mult_label_out);

\end{tikzpicture}}\!, \hbox{\begin{tikzpicture} [scale=1,transform shape] 

\def\deltax{0.3} 
\def\deltay{0.5} 


\node (mult_label_outl) at (-\deltax,+\deltay) {};
\node (mult_label_outr) at (+\deltax,+\deltay) {};
\node [dot, fill=\groupStructColour] (mult) at (0,0) {};
\node (mult_label_in) at (0,-\deltay) {};
\draw[-] [in=270,out=135] (mult) to (mult_label_outl);
\draw[-] [in=270,out=45] (mult) to (mult_label_outr);
\draw[-] (mult_label_in) to (mult);

\end{tikzpicture}}\!, \hbox{\begin{tikzpicture} [scale=1,transform shape] 

\def\deltax{0.3} 
\def\deltay{0.5} 


\node [dot, fill=\groupStructColour] (mult) at (0,0) {};
\node (mult_label_in) at (0,-\deltay) {};
\draw[-] (mult_label_in) to (mult);

\end{tikzpicture}}\!)$ and $(\hbox{\begin{tikzpicture} [scale=1,transform shape] 

\def\deltax{0.3} 
\def\deltay{0.5} 


\node (mult_label_outl) at (-\deltax,+\deltay) {};
\node (mult_label_outr) at (+\deltax,+\deltay) {};
\node [dot, fill=\classicalStructColour] (mult) at (0,0) {};
\node (mult_label_in) at (0,-\deltay) {};
\draw[-] [in=270,out=135] (mult) to (mult_label_outl);
\draw[-] [in=270,out=45] (mult) to (mult_label_outr);
\draw[-] (mult_label_in) to (mult);

\end{tikzpicture}}\!, \hbox{\begin{tikzpicture} [scale=1,transform shape] 

\def\deltax{0.3} 
\def\deltay{0.5} 


\node [dot, fill=\classicalStructColour] (mult) at (0,0) {};
\node (mult_label_in) at (0,-\deltay) {};
\draw[-] (mult_label_in) to (mult);

\end{tikzpicture}}\!, \hbox{\begin{tikzpicture} [scale=1,transform shape] 

\def\deltax{0.3} 
\def\deltay{0.5} 


\node (mult_label_inl) at (-\deltax,-\deltay) {};
\node (mult_label_inr) at (+\deltax,-\deltay) {};
\node [dot, fill=\classicalStructColour] (mult) at (0,0) {};
\node (mult_label_out) at (0,+\deltay) {};
\draw[-] [out=90,in=225](mult_label_inl) to (mult);
\draw[-] [out=90,in=315](mult_label_inr) to (mult);
\draw[-] (mult) to (mult_label_out);

\end{tikzpicture}}\!, \hbox{\begin{tikzpicture} [scale=1,transform shape] 

\def\deltax{0.3} 
\def\deltay{0.5} 


\node [dot, fill=\classicalStructColour] (mult) at (0,0) {};
\node (mult_label_out) at (0,+\deltay) {};
\draw[-] (mult) to (mult_label_out);

\end{tikzpicture}}\!)$  is \textbf{strongly complementary} if it satisfies the following \textbf{bialgebra equation} (\ref{eqn_bialgebraEqns}) and \textbf{coherence equations} (\ref{eqn_unitCopyEqns}):
		\begin{equation}\label{eqn_bialgebraEqns}
			\hbox{\begin{tikzpicture}[node distance=7.5mm]

\node (center) {};

\node (algebraTop) [smalldot, fill = \classicalStructColour]   
	[above of = center, yshift = -5mm]{};
\node (Hout) [above of = algebraTop, xshift = -5mm] {};
\node (Tout) [above of = algebraTop, xshift = +5mm] {};

\node (algebraBot) [smalldot, fill = \groupStructColour]  
	[below of = center, yshift = +5mm]{};
\node (Hin) [below of = algebraBot, xshift = -5mm] {};
\node (Tin) [below of = algebraBot, xshift = +5mm] {};

\begin{pgfonlayer}{background}
\draw[-,out=90,in=270] (algebraBot) to (algebraTop);
\draw[-,out=135,in=270] (algebraTop) to (Hout);
\draw[-,out=45,in=270] (algebraTop) to (Tout);
\draw[-,out=90,in=225] (Hin) to (algebraBot);
\draw[-,out=90,in=315] (Tin) to (algebraBot);
\end{pgfonlayer}

\node (equals) [right of = center, xshift = 0mm]{$=$};

\node (center) [right of = equals, xshift = 0mm] {};

\node (algebraTop) [smalldot, fill = \groupStructColour]   
	[above of = center, yshift = -5mm]{};
\node (Hout) [above of = algebraTop, xshift = 0mm] {};
\node (timemult) [smalldot, fill=\groupStructColour] 
	[right of = algebraTop, xshift = 0mm] {}; 
\node (Tout) [above of = timemult, xshift = 0mm] {};

\node (algebraBot) [smalldot, fill = \classicalStructColour]  
	[below of = center, yshift = +5mm]{};
\node (Hin) [below of = algebraBot, xshift = 0mm] {};
\node (timediag) [smalldot, fill=\classicalStructColour] 
	[right of = algebraBot, xshift = 0mm] {}; 
\node (Tin) [below of = timediag, xshift = 0mm] {};

\begin{pgfonlayer}{background}
\draw[-,out=135,in=225] (algebraBot) to (algebraTop);
\draw[-,out=90,in=270] (algebraTop) to (Hout);
\draw[-,out=90,in=270] (Hin) to (algebraBot);
\draw[-,out=90,in=270] (Tin) to (timediag);
\draw[-,out=90,in=270] (timemult) to (Tout);
\draw[-,out=45,in=315] (timediag) to (timemult);
\draw[-,out=135,in=315] (timediag) to (algebraTop);
\draw[-,out=45,in=225] (algebraBot) to (timemult);
\end{pgfonlayer}

\end{tikzpicture}}
		\end{equation}
		\begin{equation}\label{eqn_unitCopyEqns}
			\hbox{\begin{tikzpicture}[node distance=8mm]

\node (center) {};

\node (algebraTop) [smalldot, fill = \classicalStructColour]   
        [above of = center, yshift = -4mm]{};
\node (Hout) [above of = algebraTop, xshift = -3mm,yshift = 0mm] {};
\node (Tout) [above of = algebraTop, xshift = +3mm,yshift = 0mm] {};

\node (algebraBot) [smalldot, fill = \groupStructColour]  
        [below of = center, yshift = +8mm]{};

\begin{pgfonlayer}{background}
\draw[-,out=90,in=270] (algebraBot) to (algebraTop);
\draw[-,out=135,in=270] (algebraTop) to (Hout);
\draw[-,out=45,in=270] (algebraTop) to (Tout);
\end{pgfonlayer}

\node (equals) [right of = center, xshift = -1mm,yshift = 5mm]{$=$};

\node (center) [right of = equals, xshift = -3mm,yshift = -5mm] {};

\node (algebraTop) [smalldot, fill = \groupStructColour]   
        [below of = center, yshift = +10mm]{};
\node (Hout) [above of = algebraTop, yshift = +2mm] {};
\node (timemult) [smalldot, fill=\groupStructColour] 
        [right of = algebraTop, xshift = -4mm] {}; 
\node (Tout) [above of = timemult, yshift = +2mm] {};

\begin{pgfonlayer}{background}
\draw[-,out=90,in=270] (algebraTop) to (Hout);
\draw[-,out=90,in=270] (timemult) to (Tout);
\end{pgfonlayer}

\node (spacer) [right of = center, xshift = 3mm,yshift = 5mm]{};

\node (center) [right of = spacer, yshift = -5mm] {};

\node (algebraTop) [smalldot, fill = \groupStructColour]   
        [below of = center, yshift = +14mm]{};
\node (Hout) [below of = algebraTop, xshift = -3mm] {};
\node (Tout) [below of = algebraTop, xshift = +3mm] {};

\node (algebraBot) [smalldot, fill = \classicalStructColour]  
        [above of = center, yshift = +2mm]{};

\begin{pgfonlayer}{background}
\draw[-,out=270,in=90] (algebraBot) to (algebraTop);
\draw[-,out=225,in=90] (algebraTop) to (Hout);
\draw[-,out=315,in=90] (algebraTop) to (Tout);
\end{pgfonlayer}

\node (equals) [right of = center, xshift = -1mm,yshift = 5mm]{$=$};

\node (center) [right of = equals, xshift = -3mm, yshift = -7mm] {};

\node (algebraTop)  
        [below of = center, yshift = +8mm]{};
\node (Hout) [smalldot, fill = \classicalStructColour]  
        [above of = algebraTop, yshift = 2mm] {};
\node (timemult) 
        [right of = algebraTop, xshift = -4mm] {}; 
\node (Tout) [smalldot, fill=\classicalStructColour] 
        [above of = timemult, yshift = 2mm] {};

\begin{pgfonlayer}{background}
\draw[-,out=90,in=270] (algebraTop) to (Hout);
\draw[-,out=90,in=270] (timemult) to (Tout);
\end{pgfonlayer}

\end{tikzpicture}}
		\end{equation}
	\end{definition}
	\noindent From now on we shall refer to the structures by their colour, i.e. by $\hbox{\begin{tikzpicture} [scale=1,transform shape] 

\def\deltax{0.3} 
\def\deltay{0.5} 


\node [dot, fill=\classicalStructColour] (mult) at (0,0) {};

\end{tikzpicture}}\!$ and $\hbox{\begin{tikzpicture} [scale=1,transform shape] 

\def\deltax{0.3} 
\def\deltay{0.5} 


\node [dot, fill=\groupStructColour] (mult) at (0,0) {};

\end{tikzpicture}}\!$. A more familiar presentation of strongly complementary pairs can be given by observing that they correspond (when both structures have enough classical points to form a basis) to pairs of non-degenerate observables obeying the finite-dimensional Weyl form of the Canonical Commutation Relations \cite{StefanoGogioso-CategoricalSemanticsSchrodingersEqn}. Also, we have the following characterisation of strong complementarity in terms of group actions on classical points. 

	\begin{theorem}
		Let $\hbox{\begin{tikzpicture} [scale=1,transform shape] 

\def\deltax{0.3} 
\def\deltay{0.5} 


\node [dot, fill=\groupStructColour] (mult) at (0,0) {};

\end{tikzpicture}}\!$ and $\hbox{\begin{tikzpicture} [scale=1,transform shape] 

\def\deltax{0.3} 
\def\deltay{0.5} 


\node [dot, fill=\classicalStructColour] (mult) at (0,0) {};

\end{tikzpicture}}\!$ be a pair of $\dagger$-qSFAs. If the pair is strongly complementary, then $(\ZmultSym,\ZunitSym)$ \vspace{-3pt} acts as a group on the classical points of $\hbox{\begin{tikzpicture} [scale=1,transform shape] 

\def\deltax{0.3} 
\def\deltay{0.5} 


\node [dot, fill=\groupStructColour] (mult) at (0,0) {};

\end{tikzpicture}}\!$. We denote this group as K$_{\hbox{\begin{tikzpicture} [scale=1,transform shape] 

\def\deltax{0.3} 
\def\deltay{0.5} 


\node [dot, fill=\groupStructColour] (mult) at (0,0) {};

\end{tikzpicture}}\!}$ Conversely, if the $\hbox{\begin{tikzpicture} [scale=1,transform shape] 

\def\deltax{0.3} 
\def\deltay{0.5} 


\node [dot, fill=\groupStructColour] (mult) at (0,0) {};

\end{tikzpicture}}\!$-classical points form a basis and $(\ZmultSym,\ZunitSym)$ acts as a group on them, then the pair is strongly complementary.
	\end{theorem}
	\begin{proof}
		See \cite{StefanoGogioso-RepTheoryCQM}.
	\end{proof}

    Phases are the other fundamental ingredient of Mermin measurements. 
	\begin{definition}\label{def_Phases} 
		A \textbf{phase state} for a $\dagger$-qSCFA $\hbox{\begin{tikzpicture} [scale=1,transform shape] 

\def\deltax{0.3} 
\def\deltay{0.5} 


\node [dot, fill=\classicalStructColour] (mult) at (0,0) {};

\end{tikzpicture}}\!$ is a pure state $\ket{\alpha}$ such that: 
		\begin{equation}\label{eqn_Zphasestate}
			\hbox{\begin{tikzpicture}[node distance=8mm]

\node (center) [smalldot, fill = \Zcolour] {};

\node (map) [right of = center, xshift = -4mm] {};
\node (out) [above of = center, yshift = -2mm] {};
\node (in) [kpoint] [below of = map, yshift = +1mm] {$\alpha$};

\node (mapconj) [left of = center, xshift = +4mm] {};
\node (inconj) [kpointconj] [below of = mapconj, yshift = +1mm] {$\alpha$};

\begin{pgfonlayer}{background}
\draw[->-=.5,out=90,in=315] (in) to (center);
\draw[->-=.5,out=90,in=270] (center) to (out);
\draw[-<-=.5,out=90,in=225] (inconj) to (center);
\end{pgfonlayer}

\node (equals) [right of = center, xshift = 5mm]{$=$};

\node (center) [right of = equals, xshift = 0mm]{};
\node (unit) [smalldot, fill = \Zcolour] [below of = center, yshift = +1mm] {};

\node (out) [above of = center, yshift = -2mm] {};

\begin{pgfonlayer}{background}
\draw[->-=.5,out=90,in=270] (unit) to (out);
\end{pgfonlayer}

\end{tikzpicture}





}
		\end{equation}
		A \textbf{phase} is a map in the following form, where $\ket{\alpha}$ is a phase state for $(\hbox{\begin{tikzpicture} [scale=1,transform shape] 

\def\deltax{0.3} 
\def\deltay{0.5} 


\node [dot, fill=\classicalStructColour] (mult) at (0,0) {};

\end{tikzpicture}}\!)$:
		\begin{equation}\label{eqn_Zphase}
		\hbox{\begin{tikzpicture}[node distance = 10mm,scale=1.5]
                \node [dot,  fill=\classicalStructColour, inner sep=1] (0) at (-1, -0) {$\alpha$};
                \node (1) at (-1, -1) {};
                \node (2) at (-1, 1) {};
                \node (3) at (0, -0) {$:=$};
                \node [smalldot,  fill=\classicalStructColour] (4) at (1, -0) {};
                \node (5) at (1, -1) {};
                \node (6) at (1, 1) {};
                \node [kpoint] (7) at (2, -0.5) {$\alpha$};

        \begin{pgfonlayer}{background}
                \draw [->-=.5, bend right=45, looseness=0.75] (7) to (4);
                \draw [->-=.5] (5.center) to (4);
                \draw [->-=.5] (4) to (6.center);
                \draw [->-=.5] (1.center) to (0);
                \draw [->-=.5] (0) to (2.center);
        \end{pgfonlayer}
\end{tikzpicture}}
		\end{equation}
	\end{definition}
	In particular, elements of $K_{\hbox{\begin{tikzpicture} [scale=1,transform shape] 

\def\deltax{0.3} 
\def\deltay{0.5} 


\node [dot, fill=\groupStructColour] (mult) at (0,0) {};

\end{tikzpicture}}\!}$ are phase states, as Theorem \ref{thm_PhaseGroup} explains.

\end{document}